\documentclass[%
 aip,
 amsmath,amssymb,
 reprint,%
]{revtex4-1}

\newif\iflong
\longtrue

\iflong
\newcommand{\refappendixthm}[2]{\Cref{#1} of~\Cref{sec:derivations}}
\newcommand{\refappendixfirstoccurrence}[2]{\refappendix{#1}{#2}}
\newcommand{\refappendix}[2]{\Cref{#1}}
\else
\newcommand{\refappendixthm}[2]{\refappendix{#1}{#2}}
\newcommand{\refappendixfirstoccurrence}[2]{the supplementary information (SI) in Section~{#2}}
\newcommand{\refappendix}[2]{Section~{#2} of the SI}
\fi

\usepackage{graphicx}%
\usepackage{dcolumn}%
\usepackage{bm}%

\usepackage[utf8]{inputenc}
\usepackage[T1]{fontenc}
\usepackage{mathptmx}
\usepackage{etoolbox}

\usepackage[version=3]{mhchem} %
\usepackage{xcolor}
\usepackage{amsmath}
\usepackage{amsthm}
\usepackage{physics}
\usepackage[caption=false]{subfig}
\usepackage[ruled,vlined,linesnumbered]{algorithm2e}

\providecommand{\DontPrintSemicolon}{\dontprintsemicolon}
\usepackage{tikz}
\usepackage{xcolor}
\usepackage[export]{adjustbox}
\usepackage{paralist}
\usepackage{hyperref}
\usepackage{cleveref}
\usepackage{float}
\usepackage{comment}
\usepackage{booktabs}

\newif\ifmaintextcompiled
\maintextcompiledtrue

\newcommand{\rememberlines}{\xdef\rememberedlines{\number\value{AlgoLine}}}
\newcommand{\resumenumbering}{\setcounter{AlgoLine}{\rememberedlines}}

\makeatletter
\def\@email#1#2{%
 \endgroup
 \patchcmd{\titleblock@produce}
  {\frontmatter@RRAPformat}
  {\frontmatter@RRAPformat{\produce@RRAP{*#1\href{mailto:#2}{#2}}}\frontmatter@RRAPformat}
  {}{}
}%

\newcommand{\alghidebottomrule}{\renewcommand{\@algocf@post@ruled}}

\newcommand{\alghidetoprule}{\renewcommand{\@algocf@pre@ruled}}
\makeatother

\begin{document}

\title{Periodic Boundary Conditions for Bosonic Path Integral Molecular Dynamics}
\author{Jacob Higer}
\altaffiliation{These authors contributed equally to this work.}
\affiliation {School of Physics, Tel Aviv University, Tel Aviv 6997801, Israel.}
\author{Yotam M. Y. Feldman}%
\altaffiliation{These authors contributed equally to this work.}
\affiliation
{School of Chemistry, Tel Aviv University, Tel Aviv 6997801, Israel.}
\author{Barak Hirshberg}%
\email{hirshb@tauex.tau.ac.il}%
\affiliation{School of Chemistry, Tel Aviv University, Tel Aviv 6997801, Israel.}%
\affiliation{The Ratner Center for Single Molecule Science, Tel Aviv University, Tel Aviv 6997801, Israel.}
\affiliation{The Center for Computational Molecular and Materials Science, Tel Aviv University, Tel Aviv 6997801, Israel.}%

\newcommand{\TODO}[1]{\textcolor{red}{TODO: {#1}}}
\newcommand{\yotam}[1]{\textcolor{blue}{\bf {#1}}}
\newcommand{\jacob}[1]{\textcolor{magenta}{\bf {#1}}}
\newcommand{\notice}[1]{\textcolor{red}{\bf {#1}}}
\newcommand{\barak}[1]{\textcolor{purple}{\bf {#1}}}
\newcommand{\yotamsmall}[1]{{\footnotesize\color{magenta}[{\bf Yotam}: #1]}}
\newcommand{\alternative}[1]{\textcolor{magenta}{#1}}

\newcommand{\set}[1]{\{{#1}\}}
\newcommand{\card}[1]{\left|{#1}\right|}
\newcommand{\sumcond}[2]{\substack{{#1}, \\ {#2}}}
\newcommand{\bigO}{\mathcal{O}}

\newcommand{\fact}[1]{{#1}!}
\newcommand{\pfact}[1]{\fact{\left(#1\right)}}
\newcommand{\Symset}[1]{\mathcal{S}({#1})}
\newcommand{\Symfromto}[2]{\mathcal{S}[{#1},{#2}]}
\newcommand{\SymN}[1]{\mathcal{S}[1,{#1}]}

\newcommand{\permsubgroup}[3]{\mathcal{S}_#1 \left(#2 \rightsquigarrow #3\right)}

\newcommand{\cycleof}[2]{c_{#1}({#2})}
\newcommand{\nextof}[3]{{#3}({#1})={#2}}
\newcommand{\prevparticle}[1]{{#1}^{-}}
\newcommand{\nextparticle}[1]{{#1}^{+}}
\newcommand{\particleup}{v}
\newcommand{\particledown}{u}
\newcommand{\cyclenotate}[1]{({#1})}

\newcommand{\pos}{{\bf R}}
\newcommand{\posbead}[2]{{\bf r}_{#1}^{#2}}
\newcommand{\beadpos}[2]{\posbead{#1}{#2}}
\newcommand{\mass}{m}

\newcommand{\springfrequency}{\omega_P}
\newcommand{\springconstant}{\mass \springfrequency^2}
\newcommand{\springenergyprefix}{\frac{1}{2} \springconstant}

\newcommand{\rdiffsquared}[4]{\left(\beadpos{#1}{#2} - \beadpos{#3}{#4}\right)^2}
\newcommand{\interparticleforce}[1]{\springenergyprefix \sum_{j=1}^{P-1}{\rdiffsquared{#1}{j+1}{#1}{j}}}

\newcommand{\repsym}{G}
\newcommand{\rep}[1]{\repsym[{#1}]}

\newcommand{\boltzmann}[1]{e^{-\beta {#1}}}

\newcommand{\Vfromto}[2]{V_{\text{PBC}}^{[{#1},{#2}]}}
\newcommand{\Vfrom}[1]{\Vfromto{#1}{N}}
\newcommand{\Vto}[1]{\Vfromto{1}{#1}}
\newcommand{\Vall}{\Vfromto{1}{N}}

\newcommand{\originalpotential}[1]{\Vto{#1}}
\newcommand{\originalpotentialnotation}[1]{V_B^{({#1})}}

\newcommand{\Eperm}[1]{E^{#1}}
\newcommand{\Efromto}[2]{E_{\text{PBC}}^{[{#1},{#2}]}}
\newcommand{\Enk}[2]{\Efromto{{#1}-{#2}+1}{{#1}}}
\newcommand{\Einterior}[1]{E_{\textit{int}}^{({#1})}}
\newcommand{\permutationthree}[3]{\left(\begin{smallmatrix} 
1 & 2 & 3\\
{#1} & {#2} & {#3}
\end{smallmatrix}\right)}

\newcommand{\snip}[2]{\rep{#1} \textit{ snips at } {#2}}
\newcommand{\fullsnip}[3]{\rep{#1} \textit{ snip from } {#2} \textit{ to } {#3}}
\newcommand{\project}[2]{{#1}_{| {#2}}}

\newcommand{\beadderive}[3]{\beadderiveexplicit{#1}{#2}{#3}}
\newcommand{\beadderiveexplicit}[3]{\nabla_{\beadpos{#1}{#2}} {#3}}
\newcommand{\beadforce}[3]{- \beadderive{#1}{#2}{#3}}

\newcommand{\Prrep}{\Prerepperm{\sigma}}
\newcommand{\Prerepperm}[1]{\Pr(\rep{#1})}
\newcommand{\Prrepnext}[2]{\Pr\left(\nextof{#1}{#2}{\rep{\sigma}}\right)}

\newcommand{\windingcap}{\mathcal{W}}
\newcommand{\windingcapvector}{\windingcap}
\newcommand{\winding}[2]{\mathbf{w}_{#1}^{#2}}
\newcommand{\rdiffsquaredwinding}[4]{\left(\beadpos{#1}{#2} + \winding{#1}{#2} L - \beadpos{#3}{#4}\right)^2}
\newcommand{\windingsumindex}[2]{\substack{{\winding{#2}{1}=-{#1},\ldots,{#1}} \\ {\ldots} \\ {\winding{#2}{P}=-{#1},\ldots,{#1}}}}
\newcommand{\windingsequence}{\set{\mathbf{w}}}
\newcommand{\windingsequencefromto}[2]{\set{\mathbf{w}}^{[{#1},{#2}]}}
\newcommand{\rdiffsquaredwindingcrazy}[4]{\left(\beadpos{#1}{#2} \overset{\windingcap}{-} \beadpos{#3}{#4}\right)^2}
\newcommand{\localspringenergysym}{\mu}
\newcommand{\edgewinding}[4]{\localspringenergysym(\beadpos{#1}{#2},\beadpos{#3}{#4})} %
\newcommand{\edgewindingspecificdistinguishable}[2]{\edgewindingspecific{#1}{#2}{#1}{{#2}+1}}
\newcommand{\edgewindingspecific}[4]{\localspringenergysym(\beadpos{#1}{#2},\beadpos{#3}{#4},\winding{#1}{#2})}
\newcommand{\edgewindingdistinguishable}[2]{\sum_{\winding{#1}{#2}}{\edgewindingspecificdistinguishable{#1}{#2}}}
\newcommand{\Prereppermwinding}[2]{\Pr(\rep{#1}, {#2})}
\newcommand{\Epermdistinguishable}[1]{\Eperm{#1}}
\newcommand{\Prwinding}[2]{\Pr\left(\winding{#1}{#2}\right)}
\newcommand{\Prrepandwinding}[3]{\Pr\left(\nextof{#1}{#2}{\rep{\sigma}}, \winding{#1}{#3}\right)}
\newcommand{\Prcond}[2]{\Pr\left(#1 \mid #2\right)}
\newcommand{\Prwindingcond}[3]{\Prcond{\winding{#1}{#3}}{\nextof{#1}{#2}{\rep{\sigma}}}}
\newcommand{\windingmic}[2]{(\winding{#1}{#2})^{\ast}}
\newcommand{\Prwindingmic}[2]{\Pr\left(\windingmic{#1}{#2}\right)}

\newcommand{\efromtodbeta}[2]{A^{[{#1}, {#2}]}}
\newcommand{\Einteriordbeta}[1]{A_{\textit{int}}^{({#1})}}

\newcommand{\Prwindingext}[4]{
\frac{
\edgewindingspecific{#1}{#2}{#3}{#4}
}
    {\edgewinding{#1}{#2}{#3}{#4}}
}
\newcommand{\Prwindingextshort}[4]{
\frac{\edgewindingspecific{#1}{#2}{#3}{#4}}
    {
    \left(
    \sum_{\winding{#1}{#2}=-\windingcap}^{\windingcap}{\edgewindingspecific{#1}{#2}{#3}{#4}}
    \right)
    }
}
\newcommand{\Prwindingextshortdistinguishable}[2]{
\frac{\edgewindingspecificdistinguishable{#1}{#2}}
    {\edgewinding{#1}{#2}{#1}{{#2}+1}}
}
\newcommand{\Epermorig}[1]{\Eperm{#1}}
\newcommand{\Vtoorig}[1]{V^{[1,{#1}]}}
\newcommand{\Efromtoorig}[2]{E^{[{#1},{#2}]}}
\newcommand{\Enkorig}[2]{\Efromtoorig{{#1}-{#2}+1}{{#1}}}

\newcommand{\beadposdim}[3]{({r}_{#1}^{#2})_i}
\newcommand{\posbeaddim}[3]{\beadposdim{#1}{#2}{#3}}
\newcommand{\windingcomponent}[3]{\left({w}_{#1}^{#2}\right)_{#3}}
\newcommand{\windingdim}[3]{\windingcomponent{#1}{#2}{i}}
\newcommand{\rdiffsquaredwindingdim}[5]{\left(\beadposdim{#1}{#2}{#5} + \windingdim{#1}{#2}{#5} L - \beadposdim{#3}{#4}{#5}\right)^2}

\newcommand{\Epermidorig}{\Epermorig{\textit{id}}}
\newcommand{\Vdist}{V_{\text{D}}}
\newcommand{\Predistpermwinding}[2]{\Pr({#2})}

\newcommand{\unit}[2]{#1\;\text{#2}}
\newcommand{\density}[1]{#1\;\text{\AA}^{-3}}
\newcommand{\length}[1]{#1\;\text{\AA}}
\newcommand{\massunit}[1]{#1\;\text{u}}
\newcommand{\temperature}[1]{#1\;\mathrm{K}}
\newcommand{\dt}[1]{#1\;\text{fs}}

\newcommand{\addprime}[1]{{#1}^{\prime}}
\newcommand{\iprime}{\addprime{i}}
\newcommand{\jprime}{\addprime{j}}
\newcommand{\kprime}{\addprime{k}}
\newcommand{\uprime}{\addprime{u}}
\newcommand{\vprime}{\addprime{v}}
\newcommand{\ellprime}{\addprime{\ell}}
\newcommand{\sigmaprime}{\addprime{\sigma}}

\newcommand{\weightsym}{\mu}
\newcommand{\effweightsym}{\varphi}
\newcommand{\weight}[1]{\operatorname{\weightsym}\left(#1\right)}
\newcommand{\weightcycle}[2]{\weight{\left[#1, #2\right]}}
\newcommand{\effweight}[2]{\operatorname{\effweightsym}\left(\Symfromto{#1}{#2}\right)}

\newcommand{\Vtospecificwinding}[1]{V^{[1,{#1}]}_{\windingsequence}}
\newcommand{\Vfromspecificwinding}[1]{V^{[{#1},N]}_{\windingsequence}}
\newcommand{\Efromtospecificwinding}[2]{E_{\windingsequence}^{[{#1},{#2}]}}
\newcommand{\Enkspecificwinding}[2]{\Efromtospecificwinding{{#1}-{#2}+1}{{#1}}}
\newcommand{\Prrepnextspecificwinding}[2]{\Pr\left(\nextof{#1}{#2}{\rep{\sigma}} \, \Big{|} \, \windingsequence\right)} 
\begin{abstract}
We develop an algorithm for bosonic path integral molecular dynamics (PIMD) simulations with periodic boundary conditions (PBC) that scales quadratically with the number of particles.
Path integral methods are a powerful tool to simulate bosonic condensed phases, which exhibit fundamental physical phenomena such as Bose--Einstein condensation and superfluidity. 
Recently, we developed a quadratic scaling algorithm for bosonic PIMD, 
but employed an ad hoc treatment of PBC.
Here we rigorously enforce PBC in bosonic PIMD. It requires summing over the spring energies of all periodic images in the partition function, and a naive implementation scales exponentially with the system size.
We present an algorithm for bosonic PIMD simulations of periodic systems that scales only quadratically.
We benchmark our implementation on the free Bose gas and a model system of cold atoms in optical lattices. 
We also study an approximate treatment of PBC based on the  minimum-image convention, and derive a numerical criterion to determine when it is valid.
\end{abstract}

\maketitle

\section{Introduction}

Path integral molecular dynamics (PIMD) is a prominent method for calculating properties of quantum condensed phases at thermal equilibrium~\cite{doi:10.1063/1.446740,doi:10.1063/1.465188}.
In PIMD, thermal expectation values of the quantum system are inferred from molecular dynamics simulations of an extended classical system of ``ring polymers''. The ring polymers are formed by connecting $P$ copies of a quantum particle (``beads'' or ``imaginary time slices'') through harmonic springs whose stiffness depends on the temperature. PIMD simulations are widely used to study quantum liquids and solids~\cite{Markland2018}. 
They are also the basis for several methods for approximating quantum transport coefficients~\cite{Althorpe2021}.

Recently, we developed PIMD simulations of \emph{bosonic} systems.~\cite{hirshberg2019path,quadratic-pidmb} The challenge was that bosonic exchange symmetry required summing over an exponential number of ring polymer configurations in the partition function, each formed by connecting the rings of exchanged particles together.
We overcame this combinatorial explosion, and presented an algorithm that scales quadratically with the size of the system, $\bigO(N^2 + PN)$, where $N$ is the number of particles and $P$ is the number of beads per particle. This quadratic scaling allows efficient simulations of thousands of bosons~\cite{quadratic-pidmb}. This method is the basis also for several techniques studying fermionic systems~\cite{doi:10.1063/5.0008720,doi:10.1063/5.0030760,PhysRevE.106.025309,doi:10.1063/5.0106067,doi:10.1063/5.0093472,PhysRevE.110.065303}.

An important tool in studying quantum condensed phases is the use of \emph{periodic boundary conditions} (PBC) to capture properties of the the bulk through the simulation of a smaller system~\cite{allen2017computer}.
However, previous applications of bosonic PIMD have focused primarily on trapped, non-periodic systems, and the occasional periodic system was handled by ad hoc methods (as we explain in~\Cref{sec:mic-analysis}).
The goal of this paper is to develop an efficient, rigorous bosonic PIMD algorithm for periodic systems. %

How should the PBC modify bosonic path integral simulations?
To properly account for PBC, the partition function should account for the spring interaction of every bead with all periodic images of its neighboring beads~\cite{ceperley1995path,Dominik_Marx_1999,kleinert2009path}.
The challenge is that there are exponentially many different winding configurations that need to be taken into account in the simulation.
In path integral Monte Carlo (PIMC), the problem is addressed by sampling configurations in which the springs are stretched or compressed by integer multiples of the box length, referred to as the \emph{winding} of the springs~\cite{Cao1994,Marx1993,Miller2002,Spada2022,adrian_del_maestro_2022_7271914}.
In practice, the sampling of the windings is done up to some cutoff $\windingcap$, because configurations with highly stretched springs are energetically improbable.

The primary contribution of this paper is the development of an efficient bosonic PIMD algorithm with PBC that scales as $\bigO(\windingcap (N^2 + PN))$. While in PIMC a winding configuration for the path is sampled at each step, in PIMD it is necessary to consider all possible winding configurations simultaneously. Furthermore, 
winding and bosonic exchange are coupled, i.e., the partition function cannot be decomposed into independent sums over windings and permutations since the spring energy depends non-linearly on both.
Despite these challenges, we present an algorithm for periodic bosonic PIMD that has the same quadratic scaling with system size as the previous algorithm, which did not include summation over windings. The new algorithm also has linear scaling with the winding cutoff $\windingcap$. Thus, our algorithm rigorously extends bosonic PIMD to periodic systems.

In~\Cref{sec:background} we present the theoretical background for the paper. \Cref{sec:background-bosonic-pimd,seC:background-quadratic-pimd} summarize the theory of bosonic PIMD with quadratic scaling without accounting for the windings. \Cref{sec:background-pbc} presents the partition function with PBC, and explains the challenges in sampling it.
 \Cref{sec:results} presents the key results of the paper. First, \Cref{sec:linear-distinguishable-pimd-pbc-theory} explains how to perform PIMD simulations with PBC for the case of distinguishable particles.
 Then, we explain the theory behind the new bosonic algorithm in~\Cref{sec:quadratic-bosonic-pimd-pbc-theory}, and in~\Cref{sec:empirical}, we benchmark it on two model systems: the free Bose gas and a model of cold bosons in an optical lattice.
In~\Cref{sec:mic-analysis}, we rigorously examine a simpler approximation for PBC: applying the minimum-image convention (MIC) to the springs. We develop a quantitative criterion to when MIC can be used instead of the windings algorithm.

\section{Background}
\label{sec:background}
\subsection{Bosonic PIMD}
\label{sec:background-bosonic-pimd}
In non-periodic systems, the path integral expression for the canonical partition function of $N$ bosons at thermal equilibrium with inverse temperature $\beta$, in the presence of the physical potential $U$, is~\cite{tuckerman2010statistical}
\begin{equation}
\label{eq:int-of-sum-all-permutations}
    \mathcal{Z}^{\text B} \propto \int{
        d\pos_1 \ldots d\pos_N \, 
        \frac{1}{\fact{N}}
        \sum_{\sigma}{\boltzmann{\left(\Epermorig{\sigma} + \bar{U}\right)}}
    }.
\end{equation}
In~\Cref{eq:int-of-sum-all-permutations}, $\pos_\ell = \beadpos{\ell}{1},\ldots,\beadpos{\ell}{P}$ collectively represents the position vectors of all the $j=1,...,P$ beads of particle $\ell$, and the expression is exact in the limit $P \to \infty$. 
The sum in~\Cref{eq:int-of-sum-all-permutations} is over all permutations of the $N$ bosons.
Each permutation $\sigma$ corresponds to a ring polymer \emph{configuration} in which particles are connected according to the permutation, i.e., the last bead of particle $\ell$ is connected to the first bead of particle $\sigma(\ell)$. The spring energy of a configuration is
\begin{equation}
\label{eq:eperm-orig}
    \Epermorig{\sigma} = \springenergyprefix \sum_{\ell=1}^{N}{ \sum_{j=1}^{P}{\rdiffsquared{\ell}{j}{\ell}{j+1}}}, %
\end{equation}
with $\beadpos{\ell}{P+1} = \beadpos{\sigma(\ell)}{1}$, and $\springfrequency = {\sqrt{P}}/{(\beta \hbar)}$.
Beads  $j$ of different particles interact according to the scaled potential $\bar{U} = \frac{1}{P}\sum_{j=1}^{P}{U\left(\beadpos{1}{j},\ldots,\beadpos{N}{j}\right)}$, which is invariant under particle permutations.

\subsection{Quadratic scaling algorithm for bosonic PIMD}
\label{seC:background-quadratic-pimd}
Due to the sum over an exponential number of permutations, directly sampling the partition function through~\Cref{eq:int-of-sum-all-permutations} is computationally prohibitive.
In previous work~\cite{hirshberg2019path,quadratic-pidmb}, we showed that the same partition function can be expressed in a way that is amenable to efficient computation, by writing
\begin{equation}
\label{eq:int-of-sum-subset-permutations}
    \mathcal{Z}^{\text B} \propto \int{
        d\pos_1 \ldots d\pos_N \, 
        {\boltzmann{\left(\Vtoorig{N} + \bar{U}\right)}}
    },
\end{equation}
where the bosonic spring potential $\Vtoorig{N}$ is defined by the recurrence relation
\begin{equation}
\label{eq:our-forward-potential-recurrence-orig}
\begin{alignedat}{2}
    &\boltzmann{\Vtoorig{N}} &&= \frac{1}{N} \sum_{k=1}^{N}{\boltzmann{\left(\Vtoorig{N-k} + \Enkorig{N}{k}\right)}}.
\end{alignedat}
\end{equation}
The recurrence is terminated by $\Vtoorig{0} = 0$.
In~\Cref{eq:our-forward-potential-recurrence-orig}, $\Enkorig{N}{k}$ is the spring energy of the ring polymer that connects all the beads of particles $N-k+1,\ldots,N$ consecutively,
\begin{equation}
\Enkorig{N}{k} = \springenergyprefix \sum_{\ell=N-k+1}^{N}{\sum_{j=1}^{P}{\rdiffsquared{\ell}{j+1}{\ell}{j}}},
\label{eq:Ek-orig}
\end{equation}
where $\beadpos{\ell}{P+1}=\beadpos{\ell+1}{1}$ except $\beadpos{N}{P+1}=\beadpos{N-k+1}{1}$.

The potential $\Vtoorig{N}$ allows to perform PIMD simulations to sample the bosonic partition function in polynomial scaling~\cite{hirshberg2019path}.
We showed~\cite{quadratic-pidmb} that both the potential and the forces can be computed in quadratic time, $\bigO(N^2 + PN)$.
First, the bosonic spring \emph{potential} is computed, by evaluating the quantities $\Enkorig{N}{k}$ through another recurrence relation, extending cycles one particle at a time. Then, the spring \emph{forces} on all the beads are computed using expressions for the probabilities of the different ways particles can be connected. The full details of the previous algorithm appear in Ref.~\citenum{quadratic-pidmb}.

These equations do not rigorously address PBC. Our central achievement in this paper is developing an efficient PIMD method with quadratic scaling for periodic bosonic systems. Below, we first review the partition function with PBC, and, in~\Cref{sec:results}, derive an efficient PIMD algorithm to sample it.

\subsection{Periodic boundary conditions}
\label{sec:background-pbc}

Pollock and Ceperley~\cite{PollockCeperley1987} showed that imposing PBC alters the path integral expression for the partition function, replacing \Cref{eq:int-of-sum-all-permutations} by
\begin{equation}
\label{eq:z-bosonic-pbc-explicit}    
\mathcal{Z}^{\text B} _{\text{PBC}}
\propto \int_{D(\mathcal{V})}{
        d\pos_1 \ldots d\pos_N \, 
        \frac{1}{\fact{N}}
        \sum_{\sigma}{
            \sum_{\windingsequence}{
            \boltzmann{\left(\Eperm{\sigma,\windingsequence} + \bar{U}\right)}
        }
        }
    },
\end{equation}
where $D(\mathcal{V})$ is the spatial domain defined by the volume of the unit cell.
Throughout this paper, we assume a cubic box with side length $L$.
In~\Cref{eq:z-bosonic-pbc-explicit}, the energy of a configuration is
\begin{equation}
    \label{eq:eperm}
    \Eperm{\sigma,\windingsequence} = \springenergyprefix \sum_{\ell=1}^{N}{ \sum_{j=1}^{P}{\rdiffsquaredwinding{\ell}{j}{\ell}{j+1}}},
\end{equation}
with $\beadpos{\ell}{P+1} = \beadpos{\sigma(\ell)}{1}$.

\begin{figure*}[ht]
    \centering
    \includegraphics[width=0.75\linewidth]{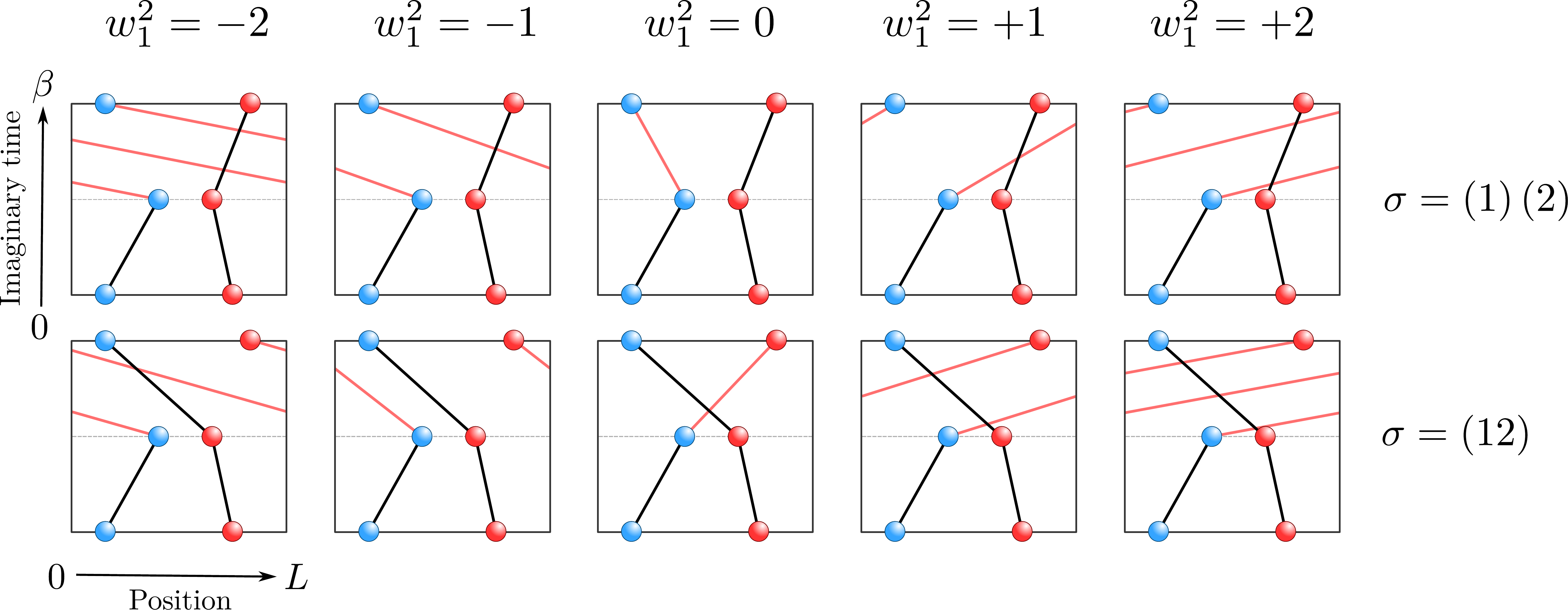}
    \caption{Winding configurations for $N=2$ particles in one-dimension with $P=2$ beads each (blue for $\posbead{1}{j}$ and red for $\posbead{2}{j}$) and winding cutoff of $\windingcap = 2$. In this case, there are only two permutations. The configurations differ by the winding of the spring connecting $\posbead{1}{P}$ to $\posbead{\sigma(1)}{1}$ (red solid line), where $\sigma$ is the permutation.
    }
    \label{fig:bosonic-wind-configuration}
\end{figure*}

\Cref{eq:z-bosonic-pbc-explicit} differs from~\Cref{eq:int-of-sum-all-permutations} in the additional sum over \emph{winding vectors}. We denote the $d$-dimensional winding vector of bead $j$ of particle $\ell$ by $\winding{\ell}{j}$. Its components are integers
expressing how many times the spring that connects a bead to its next neighbor winds around the box along a certain axis (see the top row of \Cref{fig:bosonic-wind-configuration} for an example with $d=1$). The set of \emph{windings} $\windingsequence = \winding{1}{1}, \ldots, \winding{1}{P}, \ldots, \winding{N}{1}, \ldots, \winding{N}{P}$ represents a specific combination of integer components of the winding vectors of all the beads.
Additionally, particles are connected to rings according to the permutation $\sigma$ as in~\Cref{sec:background-bosonic-pimd}.
\Cref{fig:bosonic-wind-configuration} depicts winding configurations of one of the springs for the two permutations of two bosons.

With this notation, the sum over $\windingsequence$ in~\Cref{eq:z-bosonic-pbc-explicit} %
is to be understood as
\begin{equation}
\label{eq:winding-set-sum-definition}
\sum_{\windingsequence} = \sum_{\winding{1}{1}} \dots \sum_{\winding{1}{P}} \dots \sum_{\winding{N}{1}} \dots \sum_{\winding{N}{P}},
\end{equation}
where, 
\begin{equation}
\label{eq:winding-individual-sum-definition}
\sum_{\winding{\ell}{j}} \equiv \sum_{\windingcomponent{\ell}{j}{1} = -\windingcap}^{\windingcap} \dots \sum_{\windingcomponent{\ell}{j}{d} = -\windingcap}^{\windingcap}.
\end{equation}
In practice, the sum over all possible integer components of $\winding{\ell}{j}$ in~\Cref{eq:winding-individual-sum-definition} is capped~\cite{adrian_del_maestro_2022_7271914} by introducing a non-negative integer cutoff $\windingcap$.
This is done because large windings are unfavorable as they result in large spring energies, per~\Cref{eq:eperm}.  

The challenge for PIMD is that the number of configurations in the sum of~\Cref{eq:z-bosonic-pbc-explicit} %
is exponential in both the number of particles and the number of beads, %
$N! \cdot (2\windingcap+1)^{dPN}$. %
The $N!$ term originates from particle exchange, and $(2\windingcap+1)^{dPN}$ from the sum over windings.
These two phenomena are intertwined because
the spring energy of the last bead $j=P$ of each particle $\ell$ depends non-linearly on both the winding vector $\winding{\ell}{P}$ and the permutation $\sigma(\ell)$ (see~\Cref{fig:bosonic-wind-configuration}).
For this reason, it is impossible to apply the previous algorithm to mitigate the combinatorial explosion from particle exchange, and then handle PBC separately.

In this paper, we show how to efficiently include PBC in %
bosonic PIMD. %
Our algorithm scales \emph{quadratically} with system size and only \emph{linearly} with the winding cutoff, $\bigO(\windingcap(N^2 + PN))$. Compared to the previous algorithm, the new algorithm requires changes in all its stages: the definition of the recursive bosonic potential, its evaluation, and the force calculation.
Before we explain the algorithm in~\Cref{sec:quadratic-bosonic-pimd-pbc-theory}, we first present in~\Cref{sec:linear-distinguishable-pimd-pbc-theory} an algorithm that efficiently includes PBC for the case of distinguishable particles.

\section{Results}
\label{sec:results}
\subsection{Linear scaling of distinguishable PIMD with PBC}
\label{sec:linear-distinguishable-pimd-pbc-theory}
We first describe an efficient, $\bigO(\windingcap PN)$, algorithm for distinguishable particles with PBC.
The bosonic algorithm in~\Cref{sec:quadratic-bosonic-pimd-pbc-theory} builds on the ideas presented here and extends them to include particle exchange.

The partition function for distinguishable particles with PBC is obtained by including only the identity permutation in \Cref{eq:z-bosonic-pbc-explicit},
\begin{equation}
\label{eq:z-distinguishable-pbc-explicit}    
\mathcal{Z}_{\text{PBC}}^{\text D} 
\propto \int_{D(\mathcal{V})}{
        d\pos_1 \ldots d\pos_N \, 
            \boltzmann{\left(\Vdist + \bar{U}\right)}
    },
\end{equation}
through a distinguishable spring potential that includes only summation over windings,
\begin{equation}
\label{eq:distinguishable-potential}
    \boltzmann{\Vdist} = \sum_{\windingsequence}{
    \boltzmann{
            \Epermdistinguishable{\windingsequence}
    }}.
\end{equation}
The spring energy of a configuration is then
\begin{equation}
\label{eq:dist-winding-spring-energy}
	\Epermdistinguishable{\windingsequence} = \springenergyprefix \sum_{\ell=1}^{N}{ \sum_{j=1}^{P}{\rdiffsquaredwinding{\ell}{j}{\ell}{j+1}}},
\end{equation}
and $\beadpos{\ell}{P+1} = \beadpos{\ell}{1}$. Note that PBC introduce, in principle, a sum over exponentially many winding configurations while, in the case of non-periodic distinguishable systems, there is only a single configuration.
We will now show how to perform PIMD simulations of distinguishable particles with PBC in linear time
in two steps: evaluating the potential, and evaluating the forces on all the beads.

\subsubsection{Computing the potential in \texorpdfstring{$\bigO(\windingcap PN)$}{O(WPN)} time}
We denote
the %
statistical weight of a single spring with a specific winding by
\begin{equation}
\label{eq:edgewindingspecific}
	\edgewindingspecificdistinguishable{\ell}{j} = \boltzmann{\springenergyprefix \rdiffsquaredwinding{\ell}{j}{\ell}{j+1}}.
\end{equation}
Then, the potential defined by~\Cref{eq:distinguishable-potential,eq:dist-winding-spring-energy} can be written as a sum over products of individual weights.
\begin{equation}
	\boltzmann{\Vdist} = \sum_{\windingsequence}{
		\prod_{\ell=1}^{N}{
			\prod_{j=1}^{P}{
				\edgewindingspecificdistinguishable{\ell}{j}	
			}
		}
	}.
\end{equation}

\begin{figure*}[ht]
    \centering
    \includegraphics[width=0.69\linewidth]{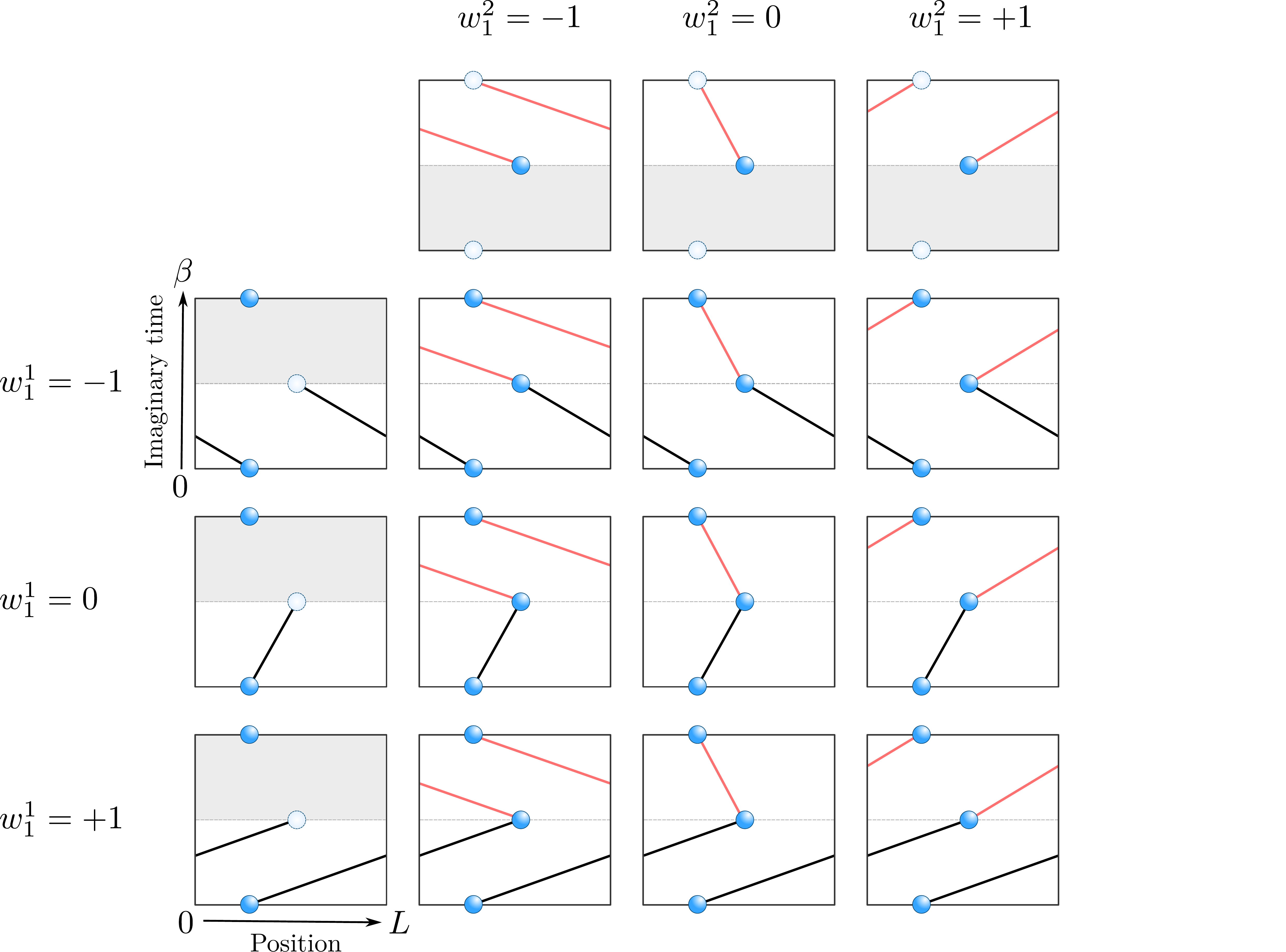}
    \caption{An illustration of calculating the distinguishable potential using \Cref{eq:distinguishable-potential-efficient} 
    in the case of a single particle in one dimension with $P=2$ beads and a winding cutoff of $\windingcap = 1$. In this case, there are $\left(2 \windingcap + 1\right)^{dNP} = 9$ different winding configurations that contribute to the potential.
    The potential is constructed by first summing over windings of individual springs, represented by the ``partial'' configurations (the contribution from shaded areas are excluded) at the top row and left column, and then multiplying their statistical weights. 
    }
    \label{fig:distinguishable-winding-configurations}
\end{figure*}

Because $\edgewindingspecificdistinguishable{\ell}{j}$ depends only on a single winding, we can rewrite the sum of products over the set $\windingsequence$ as a product of sums over an individual winding vector $\winding{\ell}{j}$ 
\begin{equation}
\label{eq:distinguishable-potential-efficient1}
    \boltzmann{\Vdist} = \prod_{\ell=1}^{N}{
                \prod_{j=1}^{P}{
                    \left(
                        \edgewindingdistinguishable{\ell}{j}
                    \right)
                }
            }.
\end{equation}
\Cref{fig:distinguishable-winding-configurations} illustrates how the total contribution of all winding vectors to $V_D$ can be computed by first summing over individual winding vectors and then multiplying the resulting contributions across different beads.
This expression shows that we evaluate the potential in $\bigO(\windingcap PN)$ time: the contribution of each of the $PN$ springs involves a sum over its windings, of which there are $2 \windingcap + 1$ possibilities. In more than one dimension, the sum  
over winding vectors of each spring in~\Cref{eq:distinguishable-potential-efficient1}
splits to the product of sums over each coordinate separately, hence the scaling is $\bigO(\windingcap)$ and not $\bigO(\windingcap^d)$---see~\refappendixfirstoccurrence{sec:multidimensional-winding-sum}{I.C} for details.
In the case of $\windingcap = 0$, the scaling of the algorithm reduces to the standard $\bigO(PN)$ distinguishable PIMD.

To simplify notation in the force derivation in the rest of the paper, 
 we denote the total weight of all the windings of a given spring by
\begin{equation}
\label{eq:edgewinding}
    \edgewinding{u}{j}{v}{k} = \sum_{\winding{u}{j}}{\edgewindingspecific{u}{j}{v}{k}},
\end{equation}
which leads to the following expression for the distinguishable spring potential
\begin{equation}
\label{eq:distinguishable-potential-efficient}
    \boltzmann{\Vdist} = \prod_{\ell=1}^{N}{
                \prod_{j=1}^{P}{
                    \edgewinding{\ell}{j}{\ell}{j+1}
                }
            }.
\end{equation}

\subsubsection{Computing the forces in \texorpdfstring{$\bigO(\windingcap PN)$}{O(WPN)} time}
We now turn to an efficient evaluation of the forces. 
To keep the reasoning as close as possible to the bosonic case, we start by taking the derivative of~\Cref{eq:distinguishable-potential} to see that 
the force can be written as a weighted average over windings:
\begin{equation}
\label{eq:distinguishable-boltzmann-configurations}
\beadforce{\ell}{j}{\Vdist} = {\sum_{\windingsequence}{\Predistpermwinding{\text{id}}{\windingsequence} \cdot \beadforce{\ell}{j}{\Epermdistinguishable{\windingsequence}}}}.
\end{equation}
In~\Cref{eq:distinguishable-boltzmann-configurations},
the force exerted on a bead in a configuration with windings $\windingsequence$ depends only on two winding vectors:
\begin{equation}
\label{eq:force-specific-configuration}
\begin{alignedat}{1}
	\beadforce{\ell}{j}{\Epermdistinguishable{\windingsequence}} = 
		&-\springconstant (\posbead{\ell}{j} - \posbead{\ell}{j-1} - \winding{\ell}{j-1}L) 
	\\
		&-\springconstant (\posbead{\ell}{j} + \winding{\ell}{j}L - \posbead{\ell}{j+1}).
\end{alignedat}
\end{equation}
Additionally, $\Predistpermwinding{\text{id}}{\windingsequence}$ is a Boltzmann probability distribution over the configurations,
\begin{equation}
\label{eq:distinguishable-representative-distribution}
    \Predistpermwinding{\text{id}}{\windingsequence} =
    \frac{
    	\boltzmann{\Epermdistinguishable{\windingsequence}}
    }{
    	\boltzmann{\Vdist}
    }.
\end{equation}

To efficiently evaluate the force in~\Cref{eq:distinguishable-boltzmann-configurations}, we combine the contributions that exert the same force, 
\begin{equation}
\label{eq:distinguiahble-bead-force-main-text}
\begin{alignedat}{1}
    \beadforce{\ell}{j}{\Vdist} = 
        &\sum_{\winding{\ell}{j-1}}{
            \Prwinding{\ell}{j-1} \cdot
            -\springconstant (\posbead{\ell}{j} - \posbead{\ell}{j-1} - \winding{\ell}{j-1}L)}
        \\
        +
        &\sum_{\winding{\ell}{j}}{
            \Prwinding{\ell}{j} \cdot
            -\springconstant (\posbead{\ell}{j} + \winding{\ell}{j}L - \posbead{\ell}{j+1})}
        .
\end{alignedat}
\end{equation}
In~\Cref{eq:distinguiahble-bead-force-main-text}, the probability $\Prwinding{\ell}{j}$ is the sum of the probabilities of all configurations in which the winding of bead $j$ of particle $\ell$ is $\winding{\ell}{j}$. %
If $\Prwinding{\ell}{j-1},\Prwinding{\ell}{j}$ are known for all $\ell$ and $j$, the force can be evaluated simply by summing over a single winding vector for each of the two spring force terms, 
in $\bigO(\windingcap)$ time, resulting in $\bigO(\windingcap PN)$ in total for all beads.

Our goal then is to evaluate the probabilities $\Prwinding{\ell}{j}$ without explicitly summing over the windings of the other beads.
Fortunately, in~\refappendixthm{thm:distinguishable-prob-si}{VI}, we show that the probability $\Predistpermwinding{\text{id}}{\windingsequence}$ decomposes into
\begin{equation}
\label{eq:distinguishable-prob-winding-independnece}
    \Predistpermwinding{\text{id}}{\windingsequence} = \prod_{\ell=1}^{N}{\prod_{j=1}^{P}{\Prwinding{\ell}{j}}},
\end{equation}
and 
\begin{equation}
\label{eq:pr-winding}
    \Prwinding{\ell}{j} = 
    \Prwindingextshortdistinguishable{\ell}{j}.
\end{equation}
Through~\Cref{eq:pr-winding}, it is possible to evaluate each probability in $\bigO(\windingcap)$ time, and all them in $\bigO(\windingcap PN)$ time in total.
As mentioned above, after this step, the forces are evaluated in additional $\bigO(\windingcap PN)$ time through~\Cref{eq:distinguiahble-bead-force-main-text}, resulting in $\bigO(\windingcap PN)$ for the evaluation of the force overall.
In more than one dimension, the sum over the windings in~\Cref{eq:distinguiahble-bead-force-main-text} splits into separate sums over the coordinates, while~\Cref{eq:pr-winding} splits into a product over the different coordinates; hence it is $\bigO(\windingcap)$ and not $\bigO(\windingcap^d)$ (see~\refappendix{sec:multidimensional-winding-sum}{I.C}).
Overall the scaling is $\bigO(\windingcap PN)$ for the force on all beads. 

Thus far, the main observation we used to reduce the scaling was the independence of windings of different springs, as apparent in~\Cref{eq:distinguishable-prob-winding-independnece} and~\Cref{fig:distinguishable-winding-configurations}.
When we turn to bosons, this is no longer the case. In each permutation separately, the same independence property of windings holds; however, the probability of a permutation depends on all windings, coupling them all. This is the main challenge of the next section.

\subsection{Quadratic scaling of bosonic PIMD with PBC}
\label{sec:quadratic-bosonic-pimd-pbc-theory}
Building on the techniques for including periodic boundary conditions in distinguishable PIMD,
we now proceed to describe an efficient algorithm for the bosonic case, achieving quadratic scaling. 

To sample the bosonic partition function of~\Cref{eq:z-bosonic-pbc-explicit}, similarly to the previous bosonic algorithm (see~\Cref{seC:background-quadratic-pimd}), we use an effective ring-polymer potential $\Vall$ defined by the recurrence relation
\begin{equation}
\label{eq:our-forward-potential-recurrence}
\begin{alignedat}{2}
    &\boltzmann{\Vall} &&= \frac{1}{N} \sum_{k=1}^{N}{\boltzmann{\left(\Vto{N-k} + \Enk{N}{k}\right)}}.
\end{alignedat}
\end{equation}
The recursion is terminated by setting $\Vto{0} = 0$.
The crucial difference between~\Cref{eq:our-forward-potential-recurrence} and~\Cref{eq:our-forward-potential-recurrence-orig}, is that the term $\Enk{N}{k}$ includes a sum over the windings,
\begin{equation}
\begin{alignedat}{1}
\boltzmann{\Efromto{\particledown}{\particleup}} &= 
\sum_{\windingsequencefromto{\particledown}{\particleup}}
\boltzmann{
        \springenergyprefix 
        \sum\limits_{\ell=\particledown}^{\particleup}{\sum\limits_{j=1}^{P}{
                \rdiffsquaredwinding{\ell}{j}{\ell}{j+1}
            }
        }
    }
\\
&=
\sum_{\windingsequencefromto{\particledown}{\particleup}}
    \prod_{\ell=\particledown}^{\particleup}{
        \prod_{j=1}^{P}{
            \edgewindingspecific{\ell}{j}{\ell}{j+1}
        }
    }, %
\label{eq:Ek}
\end{alignedat}
\end{equation}
where $\windingsequencefromto{\particledown}{\particleup}$ represents the winding vectors of all the beads of particles $\particledown,\ldots,\particleup$. %
In~\Cref{eq:Ek}, $\beadpos{\ell}{P+1} = \beadpos{\ell+1}{1}$ except $\beadpos\particleup{P+1} = \beadpos{\particledown}{1}$. Following the existing terminology~\cite{quadratic-pidmb}, we call $\Efromto{\particledown}{\particleup}$ the \emph{cycle energies}.
By the same reasoning as in distinguishable particles (\Cref{eq:distinguishable-potential-efficient1}), we can write the cycle energy as 
\begin{equation}
\begin{alignedat}{1}
\boltzmann{\Efromto{\particledown}{\particleup}} &= 
\prod_{\ell=\particledown}^{\particleup}{
    \prod_{j=1}^{P}{
        \edgewinding{\ell}{j}{\ell}{j+1}
    }
}
.
\end{alignedat}
\end{equation}

We show in~\refappendixthm{thm:bosonic-pimd-correctness}{VI} that the potential of~\Cref{eq:our-forward-potential-recurrence} leads to sampling of the correct bosonic partition function with PBC (\Cref{eq:z-bosonic-pbc-explicit}). 
The remainder of this section explains how we efficiently compute the potential $\Vall$ and the forces it induces, $\beadforce{\ell}{j}{\Vall}$ on all the beads.
We note that when $\windingcap =0$, \Cref{eq:Ek} coincides with the previous algorithm.

\subsubsection{Computing the potential in \texorpdfstring{$\bigO(\windingcap(N^2 + PN))$}{O(W(PN+NN))} time}
We start with evaluating the potential $\Vall$ (\Cref{eq:our-forward-potential-recurrence}). 
The first and most significant step is to evaluate the cycle energies.
As in the previous algorithm~\cite{quadratic-pidmb}, quadratic scaling is achieved by extending cycles one particle at a time.
This is done by adding and removing springs from the ring polymer, but we must include all the possible windings associated with the spring, which is achieved by multiplying and dividing by their statistical weights.

\label{sec:e-recursion-full}
We compute the cycle energies by
\begin{equation}
    \label{eq:e-recurrence-step}
        \begin{aligned}
        \boltzmann{\Efromto{\particledown}{\particleup}} = \boltzmann{\Efromto{\particledown+1}{\particleup}}
                            & / %
                            \edgewinding{\particleup}{P}{\particledown+1}{1}
                            \\
                            & \cdot %
                            \edgewinding{\particledown}{P}{\particledown + 1}{1}
                            \\
                            & \cdot \boltzmann{\Einterior{\particledown}}
                            \\
                            & \cdot %
                            \edgewinding{\particleup}{P}{\particledown}{1}
        \end{aligned}
\end{equation}
In~\Cref{eq:e-recurrence-step}, we first remove the contribution of all the windings of the spring that closes the cycle from the last bead of $\particleup$ to the first bead of $\particledown+1$ (see the red dashed lines in the left column of~\Cref{fig:cycle-recurrence-winding}, which depicts the three particle case).
\begin{figure*}[ht]
    \centering
    \includegraphics[width=0.45\linewidth]{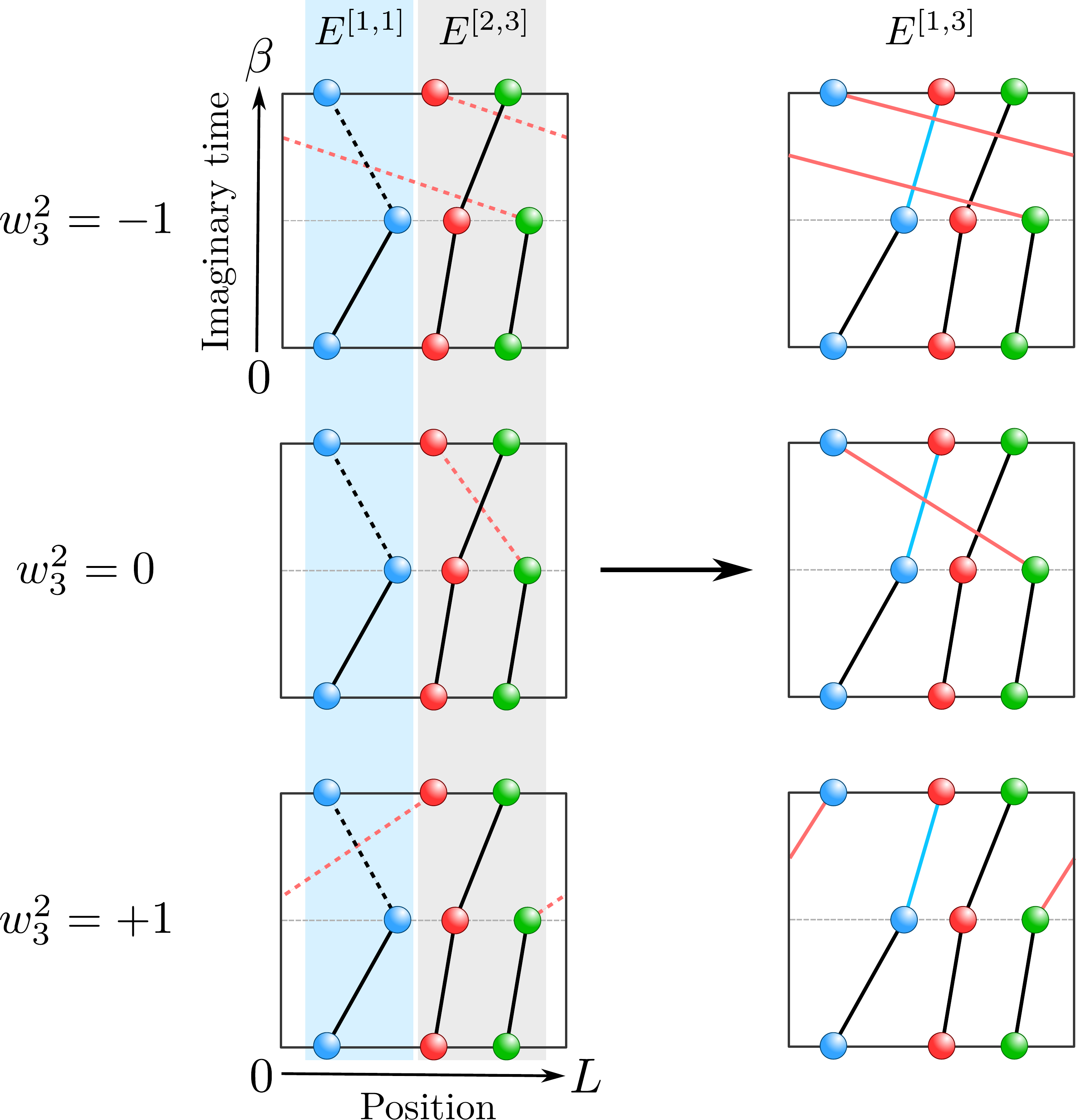}
    \caption{
    Illustration of the idea behind~\Cref{eq:e-recurrence-step}, in the case of three bosons and $P=2$, for $u=1$ and $v=3$. For simplicity, we focus only on the exterior spring when presenting the different winding configurations for a cutoff of $\windingcap = 1$. The left column depicts some of the winding configurations that contribute to $\Efromto{1}{1}$ and $\Efromto{2}{3}$. The red and black dashed lines indicate springs that must be removed. The red and cyan solid lines represent the springs that must be added. In the left column, the separate cycles are highlighted in different colors for emphasis.
    }
    \label{fig:cycle-recurrence-winding}
\end{figure*}
We then add all the windings of the spring that connects the last bead of the new particle $\particledown$ to the beginning of the previous cycle, the first bead of particle $\particledown + 1$ (corresponding to cyan solid lines in the right column of~\Cref{fig:cycle-recurrence-winding}). Then, we add the windings of all the interior springs of $\particledown$, as we explain below (not shown in~\Cref{fig:cycle-recurrence-winding}). Finally, we add all the windings of the spring that now closes the cycle from the last bead of particle $\particleup$ to the first bead of $\particledown$ (red solid line in the right column of~\Cref{fig:cycle-recurrence-winding}). 

The contribution of the windings of the interior springs, $\boltzmann{\Einterior{\particledown}}$, is defined by
\begin{equation}
\label{eq:e-interior}
    \boltzmann{\Einterior{\particledown}} = 
    \prod_{j=1}^{P-1}{\edgewinding{\particledown}{j}{\particledown}{j+1}},
\end{equation}
which includes all the springs of particle $\particledown$ that do not depend on the connectivity of $\particledown$ in the cycle.
The recurrence of~\Cref{eq:e-recurrence-step} is terminated by cycles of just one particle:
\begin{equation}
\label{eq:e-base}
    \boltzmann{\Efromto{\particledown}{\particledown}} = \edgewinding{\particledown}{1}{\particledown}{P} \cdot \boltzmann{\Einterior{\particledown}}.
\end{equation}
The manipulations performed in~\Cref{eq:e-recurrence-step} rely on the fact that it is possible to sum over the winding vectors of springs of each cycle separately. 

With these expressions, evaluating all the cycle energies requires $\bigO(\windingcap(N^2 + PN))$. The most time consuming part is evaluating~\Cref{eq:edgewinding} for $\bigO(PN)$ internal springs and $\bigO(N^2)$ springs between the last and the first bead of each pair of particles, and each such evaluation requires $\bigO(\windingcap)$ time. As mentioned above, even in more than one dimension, the scaling is $\bigO(\windingcap)$ and not $\bigO(\windingcap^d)$ (see~\refappendix{sec:multidimensional-winding-sum}{I.C}). Once the statistical weights are known, evaluating the contribution of the interior springs according to~\Cref{eq:e-interior} takes $\bigO(P)$ time for each of the $N$ particles; evaluating \Cref{eq:e-base} is $\bigO(1)$ for each of the $N$ single-particle cycles; using \Cref{eq:e-recurrence-step} takes additional $\bigO(1)$ for each of the $\bigO(N^2)$ cycles.
After evaluating the cycle energies, evaluating the potentials $\Vto{1},\ldots,\Vto{N}$ of~\Cref{eq:our-forward-potential-recurrence} takes only additional $\bigO(N^2)$ time ($\bigO(N)$ for each potential). Thus the potential $\Vall$ is computed in $\bigO(\windingcap(N^2 + PN))$ overall.
We now turn to computing the forces, in quadratic time as well.

\subsubsection{Computing the force in \texorpdfstring{$\bigO(\windingcap(N^2 + PN))$}{O(W(PN+NN))} time}
To compute the force induced by $\Vall$, we first express the force as a \emph{weighted average} over the configurations stemming from different permutations and windings.
In~\refappendixthm{thm:bosonic-potential-sum-representatives}{VI}, we show that $\Vall$ can be written as
\begin{equation}
\label{eq:vall-sum-representatives}
    \boltzmann{\Vall} = 
        \frac{1}{\fact{N}} \sum_{\sigma}{
            \sum_{\windingsequence}{
            \boltzmann{\Eperm{\rep{\sigma}, \windingsequence}}
        }
        },
\end{equation}
and that this is equivalent to~\Cref{eq:our-forward-potential-recurrence}.

This is similar to the previous bosonic algorithm, except that the choice of winding vectors generates more configurations for every permutation. In~\Cref{eq:vall-sum-representatives}, the function $\repsym$ replaces some permutations by others, because, as in the previous algorithm, not all permutations directly appear in $\Vall$; the definition is the same as in the previous algorithm~\cite{quadratic-pidmb}, unaltered by PBC. %

By taking the derivative of~\Cref{eq:vall-sum-representatives}, the force can be written as a weighted average over permutations and windings,
\begin{equation}
\label{eq:boltzmann-configurations}
\beadforce{\ell}{j}{\Vall} = \sum_{\sigma}{\sum_{\windingsequence}{\Prereppermwinding{\sigma}{\windingsequence} \cdot \beadforce{\ell}{j}{\Eperm{\rep{\sigma}, \windingsequence}}}},
\end{equation}
where
\begin{equation}
\label{eq:representative-distribution}
    \Prereppermwinding{\sigma}{\windingsequence} 
    =
    \frac{
        \boltzmann{
            \Eperm{\rep{\sigma},\windingsequence}
        }
    }{
        \fact{N} \cdot \boltzmann{\Vall}
    },
\end{equation}
is a Boltzmann distribution over the configurations.
The force in a configuration $\beadforce{\ell}{j}{\Eperm{\rep{\sigma}, \windingsequence}}$ is the same as in~\Cref{eq:force-specific-configuration}, except using $\beadpos{\ell}{P + 1} = \beadpos{\sigma(\ell)}{1}$, according to the permutation.

Computing the force on a bead efficiently is based on grouping the configurations that exert the same force on the bead. As in the previous bosonic algorithm, we separate the force evaluation for the first and last beads of each particle, referred to as \emph{exterior beads}, and the rest of the beads, which are referred to as \emph{interior beads}.

\paragraph{Force on interior beads}
The force on bead $j \neq 1, P$ of particle $\ell$ is the same as in the distinguishable particle case with PBC (\Cref{eq:distinguiahble-bead-force-main-text}). The reason is that the force on an interior bead $\posbead{\ell}{j}$ in each configuration is independent of the permutation $\sigma$ and only depends on the winding vectors $\winding{\ell}{j},\winding{\ell}{j-1}$. Thus the expression reduces to a sum over winding vectors, which can be treated the same way as in distinguishable particles. These forces are evaluated in $\bigO(\windingcap PN)$ for all the $P-2$ interior beads of each of the $N$ particles.

\paragraph{Force on exterior beads}
The force on beads $1,P$ of each particle depends both on the permutation and on winding (see~\Cref{fig:bosonic-wind-configuration}). 
Our approach is similar to the distinguishable case (\Cref{eq:distinguiahble-bead-force-main-text}), grouping contributions that exert the same force on a given bead. 
Accounting for both permutations and windings leads to the following expression for bead $P$,
\begin{equation}
    \label{eq:force-on-P-sum-neighbors}
    \begin{alignedat}{1}
    &\beadforce{\ell}{P}{\Vall} = 
    \\
    & %
      \sum_{\winding{\ell}{P-1}}
        {
            \Prwinding{\ell}{P-1} \cdot
            -\springconstant (\posbead{\ell}{P} - \posbead{\ell}{P-1} - \winding{\ell}{P-1}L)}
        \\
    +
    &\sum_{\ellprime=1}^{N}{
        \sum_{\winding{\ell}{P}}
        {
        \Prrepandwinding{\ell}{\ellprime}{P} \cdot
        -\springconstant (\posbead{\ell}{P} + \winding{\ell}{P}L - \posbead{\ellprime}{1})}
        }
        .
    \end{alignedat}
\end{equation}
\Cref{eq:force-on-P-sum-neighbors} contains contributions from two springs: the spring connecting bead $P$ to the previous bead, $\beadpos{\ell}{P-1}$, and the spring connecting it to the next bead, $\beadpos{\rep{\sigma}(\ell)}{1}$, which depends on the permutation.
The first term averages over all the winding vectors for the former, and the second over all the winding vectors of the latter, as well as permutations.
An analogous expression for the force on the first bead appears in~\refappendix{sec:force-1-full}{I.A}.

The probability $\Prwinding{\ell}{P-1}$ for the spring that is not affected by exchange is the same as in~\Cref{eq:pr-winding}. The other spring, however, is affected by exchange, and so~\Cref{eq:pr-winding} does not apply. 
Instead, \Cref{eq:force-on-P-sum-neighbors} uses the \emph{joint probability} $\Prrepandwinding{\ell}{\ellprime}{P}$. It is defined as the sum of the probabilities of all the configurations where the permutation satisfies $\rep{\sigma}(\ell) = \ellprime$ \emph{and} the winding vector of bead $P$ of particle $\ell$ is equal to $\winding{\ell}{P}$. 
This joint probability can be written as a product, 
\begin{equation}
\label{eq:windingpermute-joint-prob}
\begin{alignedat}{1}
    &\Prrepandwinding{\ell}{\ellprime}{P} = 
    \\
    &\quad \Prrepnext{\ell}{\ellprime} \cdot 
    \Prwindingcond{\ell}{\ellprime}{P},
\end{alignedat}
\end{equation}
where $\Prrepnext{\ell}{\ellprime}$ is the probability that bead $P$ of particle $\ell$ is connected to the first bead of particle $\ellprime$, and $\Prwindingcond{\ell}{\ellprime}{P}$ is the conditional probability of having $\winding{\ell}{P}$ \emph{given that} $\rep{\sigma}(\ell) = \ellprime$. 
The benefit of this decomposition is that both these probabilities can be evaluated efficiently, as we now show.

\paragraph{Connection probability}
The marginal probabilities $\Prrepnext{\ell}{\ellprime}$ retain the same form as in the previous algorithm~\cite{quadratic-pidmb}, %
the only difference being that the term $\Efromto{\particledown}{\particleup}$ includes the sum over windings per~\Cref{eq:Ek}. For instance, we show in~\refappendixthm{lem:close-cycle-probability}{VI} that for $\ell' \leq \ell$,
\begin{equation}
\label{eq:connection-prob-close-cycle-main}
    \Prrepnext{\ell}{\ell'} = 
        \frac{1}{\ell} 
        \frac{\boltzmann{\left(\Vto{\ell'-1} + \Efromto{\ell'}{\ell} + \Vfrom{\ell+1}\right)}}{\boltzmann{\Vall}} ,
\end{equation}
where the partial bosonic potentials are defined through the recurrence relation
\begin{equation}
\label{eq:potentials-forward-main}
    \boltzmann{\Vfrom{u}} = \sum_{\ell=u}^{N}{\frac{1}{\ell} \boltzmann{\left(\Efromto{u}{\ell} + \Vfrom{\ell+1}\right)}},
\end{equation}
with $u=1,\dots,N$. Expressions for all the connection probabilities appear in~\refappendix{sec:connection-prob-full}{I.B}.

The derivation of~\Cref{eq:connection-prob-close-cycle-main,eq:potentials-forward-main} relies of the fact that a permutation can be decomposed into disjoint cycles. With periodic boundary conditions, the sum over permutations and windings decomposes too into a sum over disjoint cycles \emph{and} windings. This is the reason that the connection probabilities have the same expression as in the previous algorithm. In fact, this is what allows us to define the potential with PBC itself in a recursive manner (\Cref{eq:our-forward-potential-recurrence}).

\paragraph{Conditional winding probability}
The conditional probability is defined as a sum over all the configurations in which $\rep{\sigma}(\ell) = \ellprime$ and have the winding $\winding{\ell}{P}$.
We show in~\refappendixthm{lemma:winding-probability-given-connection-same-spring}{VI} that it can be calculated efficiently by
\begin{equation}
\label{eq:windingpermute-conditional-prob}
    \Prwindingcond{\ell}{\ellprime}{P} =
    \Prwindingext{\ell}{P}{\ellprime}{1}.
\end{equation}
This expression is identical to~\Cref{eq:pr-winding}, %
except that here $\beadpos{\ell}{P+1}=\beadpos{\ellprime}{1}$.
As~\Cref{eq:windingpermute-conditional-prob} shows, once the connectivity is set, all permutations have the same probability of winding $\winding{\ell}{P}$, which has the same form as in the case of distinguishable particles. %

\paragraph{Complexity of force evaluation}
Overall, the force on all the exterior beads is computed in $\bigO(\windingcap N^2)$ time: first, the connection probabilities $\Prrepnext{\ell}{\ellprime}$ are computed in $\bigO(N^2)$ according to~\Cref{eq:connection-prob-close-cycle-main} and the rest of the expressions appearing in~\refappendix{sec:connection-prob-full}{I.B}. Second, the $\windingcap N^2$ joint probabilities are computed in $\bigO(1)$ each according to~\Cref{eq:windingpermute-conditional-prob}. Finally, the force on each of the $2N$ exterior beads is computed in $\bigO(\windingcap N)$ according to~\Cref{eq:force-on-P-sum-neighbors}. The same considerations apply when evaluating the force on the first bead of each particle. As before, in the multidimensional case, the sum over windings in~\Cref{eq:force-on-P-sum-neighbors} splits into separate sums for different axes, hence it is $\bigO(\windingcap)$ and not $\bigO(\windingcap^d)$ (see~\refappendix{sec:multidimensional-winding-sum}{I.C}).

\subsubsection{Estimator for the kinetic energy in \texorpdfstring{$\bigO(\windingcap(N^2 + PN))$}{O(W(PN+NN))} time}
We derive a thermodynamic kinetic energy estimator appropriate for PBC,
which is given by
\begin{equation}
\label{eq:thermodynamics-estimator-main}
    \ev{E} = \frac{dPN}{2\beta} + \ev{\Vall +\beta\frac{\partial \Vall}{\partial\beta}} +  \ev{\bar{U}},
\end{equation}
where the brackets denote an average over the ensemble with PBC. This is the same estimator as in the previous algorithm, except that the recurrence relation for $\Vall +\beta\frac{\partial \Vall}{\partial\beta}$ must include the cycle energies with the sum over windings, $\Efromto{\particledown}{\particleup}$, and their derivative with respect to the inverse temperature $\beta$.
We provide full details on how this expression is evaluated efficiently in~\refappendix{sec:kinetic-estimator-full}{II}.

\subsubsection{Additional implementation details}
Pseudocode for our algorithm appears in~\refappendix{sec:alg-pseudocode}{III}. %
In our implementation, we also take care of the following:
\paragraph{Wrapping coordinates}
The limits of integration in~\Cref{eq:z-bosonic-pbc-explicit} restrict the coordinates of the particles to the unit cell. Therefore, in the course of the simulation, we prevent the particles from leaving the box by adding or subtracting integer multiples of the box length, whenever the particle leaves, as is customary in MD simulations with PBC~\cite{allen2017computer}.

\paragraph{Numerical stability} 
The log-sum-exp method~\cite{10.1093/imanum/draa038} is used~\cite{PhysRevE.106.025309} to ensure numerical stability whenever a sum of exponentials is used. In addition to the stages where such a sum is performed in the previous algorithm, our algorithm also includes such a sum as part of computing the cycle energies (described in~\Cref{eq:e-recurrence-step}).

\subsection{Numerical results}
\label{sec:empirical}
We applied our algorithm in PIMD simulations of two periodic systems: the free Bose gas, and particles in a sinusoidal trap. Below, we benchmark 
the new algorithm (labeled as ``\textsc{pimd-b (pbc)}'') and compare it to analytical results (labeled as ``\textsc{exact (pbc)}''), and simulations which neglect PBC using the previous algorithm (labeled as ``\textsc{pimd-b (no pbc)}''). Then, we validate the theoretical scaling of the algorithm, quadratic with $N$ and linear with $\windingcap$.

\paragraph{Benchmark}
In the free Bose gas, \Cref{fig:free-particles-energy-temperature} presents the resulting energy from our bosonic PIMD with PBC algorithm, at a range of temperatures, for $N=64$ and a density of $\density{0.035}$. A maximum of $P=32$ beads and a winding cutoff of $\windingcap = 1$ were required for convergence. Additional details appear in~\refappendix{sec:numerical-additional}{IV}. For comparison, the results of the previous algorithm are also presented.
The analytical solution for the free Bose gas with PBC is derived in~\refappendix{sec:analytical-free-particles}{V}.
We find very good agreement between our results and the analytical solution, with a relative error between 0.5\% and 3.3\%, in a temperature range where neglecting PBC leads to large errors.
\begin{figure}[ht]
    \centering
    \includegraphics[width=\linewidth]{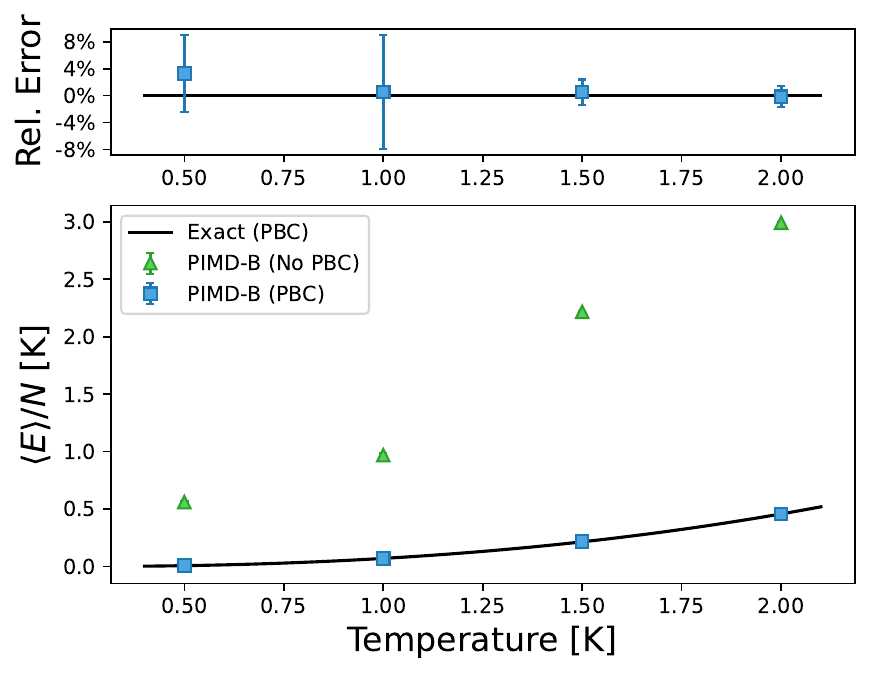}
    \caption{Energy per particle as a function of the temperature for the free Bose gas with $N=64$ atoms and density $n = \density{0.035}$. The data points correspond to bosonic PIMD simulations with PBC (blue) and without PBC (green). The solid line is the analytical result with PBC. The bottom panel shows absolute energy values, while the top panel shows the relative error of the PIMD-B PBC results when compared to the exact results.} 
    \label{fig:free-particles-energy-temperature}
\end{figure}

\begin{figure}[ht]
    \centering
    \includegraphics[width=\linewidth]{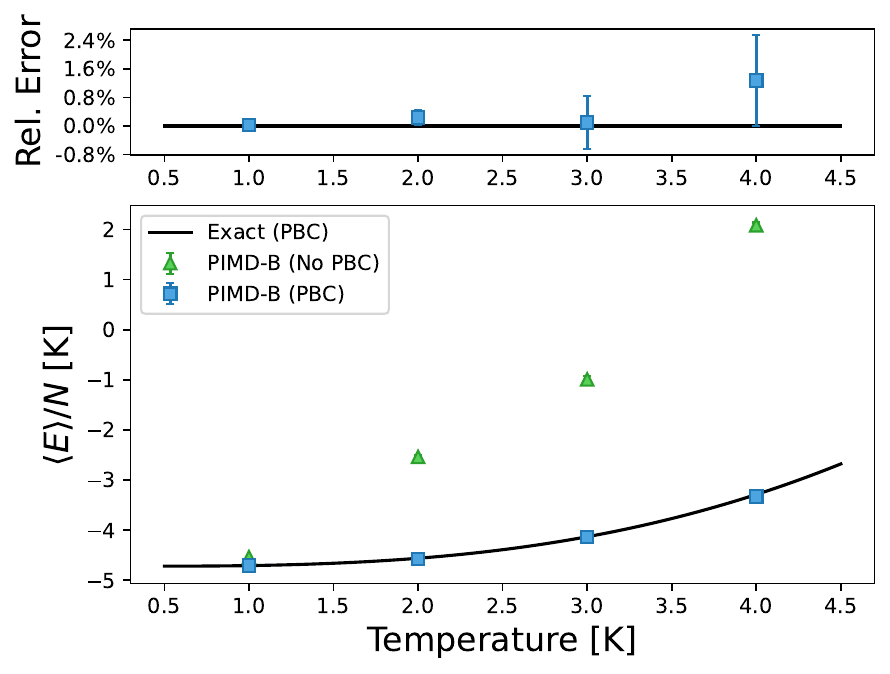}
    \caption{Energy per particle as a function of the temperature of $N=32$ particles in a sinusoidal trap. The data points correspond to bosonic PIMD simulations with PBC (blue) and without PBC (green). The solid line is the analytical result for the periodic system. The bottom panel shows absolute values, while the top panel shows the relative error of the PBC results when compared to the exact result.}
    \label{fig:cosine-trap-energy-temperature}
\end{figure}

For the periodic sinusoidal trap, \Cref{fig:cosine-trap-energy-temperature} presents the energy of a system of $N=32$ particles at the same density. Here, too, our algorithm reproduces the analytical result excellently, with a relative error ranging from $0.03 \%$ to $1.3 \%$. See~\refappendix{sec:analytical-cosine-trap}{V} for the derivation of the analytical result. For comparison, PIMD simulations with the previous algorithm are also presented, and lead to large errors.
These results demonstrate the correctness of our method of including PBC in bosonic PIMD.

\paragraph{Scaling}
\Cref{fig:scaling-size} shows the time required for a PIMD step with our bosonic algorithm as a function of $N$ in a simulation of free particles and $\windingcap=1$ (blue). The results are consistent with quadratic scaling with $N$, with a slope of $1.85$ in a log-log scale and a Pearson correlation coefficient $R^2=0.9998$.
Including PBC with $\windingcap = 1$ is roughly $\times 10$ slower in comparison to the previous algorithm (green).

\Cref{fig:scaling-wind} shows the time required for a PIMD step with our algorithm in a simulation of $N=64$ free particles and a varying winding cutoff $\windingcap$. The results are consistent with linear scaling with $\windingcap$ in a log-log scale and a Pearson correlation coefficient $R^2 = 0.9994$.
In~\Cref{fig:scaling-size,fig:scaling-wind}, the time measured in each point is averaged over $1000$ steps, the temperature is $T=\temperature{3.0}$, and $P=32$ and $P=4$ for for \Cref{fig:scaling-size} and \Cref{fig:scaling-wind}, respectively.
\begin{figure}[ht]
    \centering
    \includegraphics[width=\linewidth]{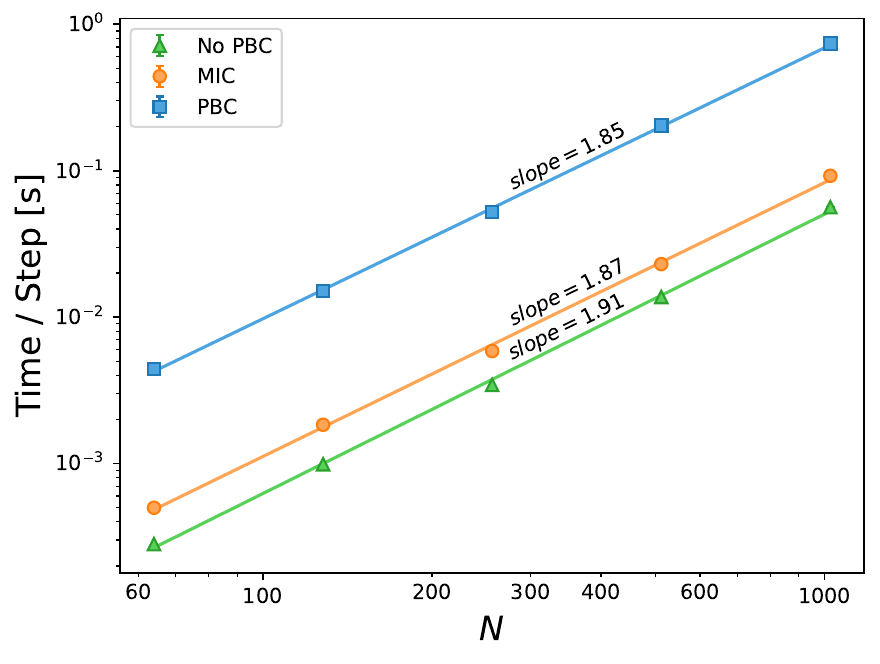}
    \caption{Scaling with the number of particles $N$ of bosonic PIMD with PBC and $\windingcap=1$ (blue). %
    For comparison, we also show the scaling of %
    bosonic PIMD with the minimum-image convention, as explained in~\Cref{sec:mic-analysis} (orange). Both algorithms are compared to the original non-periodic bosonic algorithm (green).
    In all cases,the fitted slope in log-log scale is close to $2$, as expected for quadratic scaling, but the prefactor varies.
    }
    \label{fig:scaling-size}
\end{figure}

\begin{figure}[ht]
    \centering
    \includegraphics[width=\linewidth]{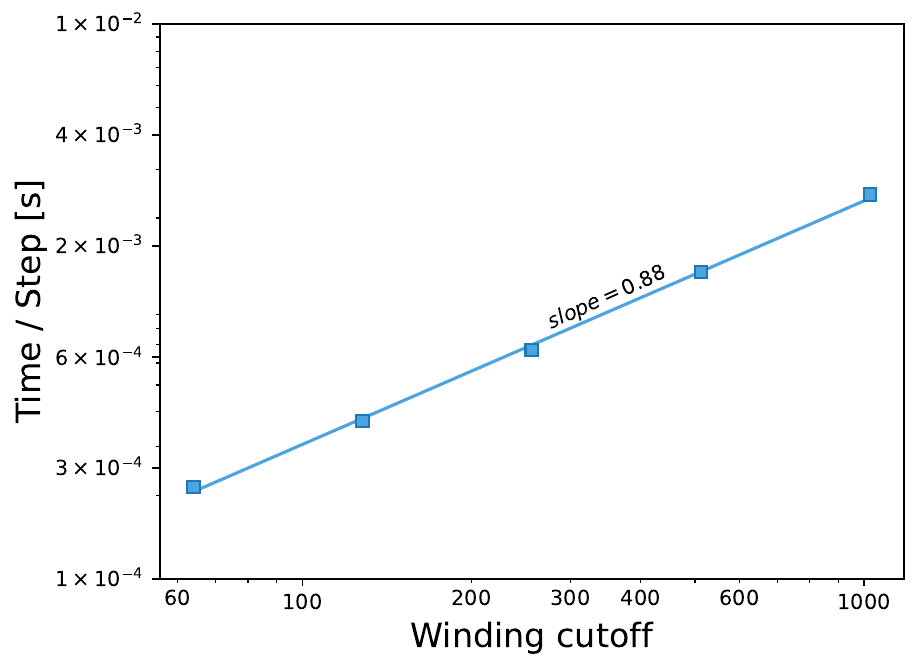}
    \caption{Scaling of bosonic PIMD with PBC as a function of the winding cutoff $\windingcap$. In log-log scale, the fitted slope is close to $1$, as expected for linear scaling.}
    \label{fig:scaling-wind}
\end{figure}

\section{PBC vs.\ the minimum image convention}
\label{sec:mic-analysis}
In this section, we compare the new algorithm for bosonic PIMD with PBC to applying the minimum-image convention (MIC) to the springs as an approximate way of including PBC in the simulation. 
Specifically, we use 
the previous algorithm without PBC, but 1) wrap coordinates inside the unit cell, and 
2) replace the differences $\beadpos{\ell}{j} - \beadpos{\ell}{j+1}$ in the spring energies and forces with the displacement to the nearest periodic image, by choosing the winding vector $\windingmic{\ell}{j}$ that minimizes it, %
\begin{equation}
\label{eq:mic-def}
    \windingmic{\ell}{j} = 
    \mathop{\text{arg min}}_{\winding{\ell}{j} \in \mathbb{Z}^d}
    \set{
        \beadpos{\ell}{j} - \beadpos{\ell}{j+1} + \winding{\ell}{j}L
    }. 
\end{equation}
In the multidimensional case each scalar component of~\Cref{eq:mic-def} is minimized separately.
In essence, the MIC selects a single winding vector per bead for which the energy is minimal, and the probability is the largest according to \Cref{eq:windingpermute-conditional-prob,eq:pr-winding}.
Although very reasonable, this approach is not rigorous, since the MIC algorithm does not take into account the exponentially-many other choices of winding vectors. 
We next investigate when this strategy works or fails in practice.

We find that in the periodic systems we consider, the MIC approximation converges to the correct result when $P$ is sufficiently large. However, we find that this $P$ can be unnecessarily large compared to the convergence of the rigorous bosonic PIMD with PBC algorithm in some cases.
We demonstrate that the difference between the algorithm with PBC and the MIC approximation disappears for $P$ that is large enough so that the distribution of winding numbers is narrowly centered around a single value. This explains why the MIC approximations converges to the correct result in practice, and why it fails to do so in values of $P$ lower than required for convergence of the PBC algorithm.

\paragraph{MIC approximation validity}
\Cref{fig:winding-mic-convergence-free} shows the convergence with $P$ of both the algorithm with PBC and the MIC approximation for the free Bose gas with $N=64$, $T=\temperature{0.5}$, $n=\density{0.035}$. If not shown, the statistical error is smaller than the symbol size. We find that MIC converges to the correct result with PBC at about $P=14$ while the rigorous treatment of PBC converges at $P=4$.
As another example, we compare the results of rigorous PBC treatment with the MIC for the sinusoidal trap. \Cref{fig:winding-mic-convergence-cosine-trap} shows the convergence with $P$ for bosons in a sinusoidal trap with $N=32$, $T=\temperature{1.0}$, $n=\density{0.035}$. There is a small difference between the PBC algorithm and the MIC approximation in this case, and both converge at about $P=26$. 

\paragraph{Criterion for using the MIC approximation} 
\begin{figure}[ht]
    \centering
    \includegraphics[width=\linewidth]{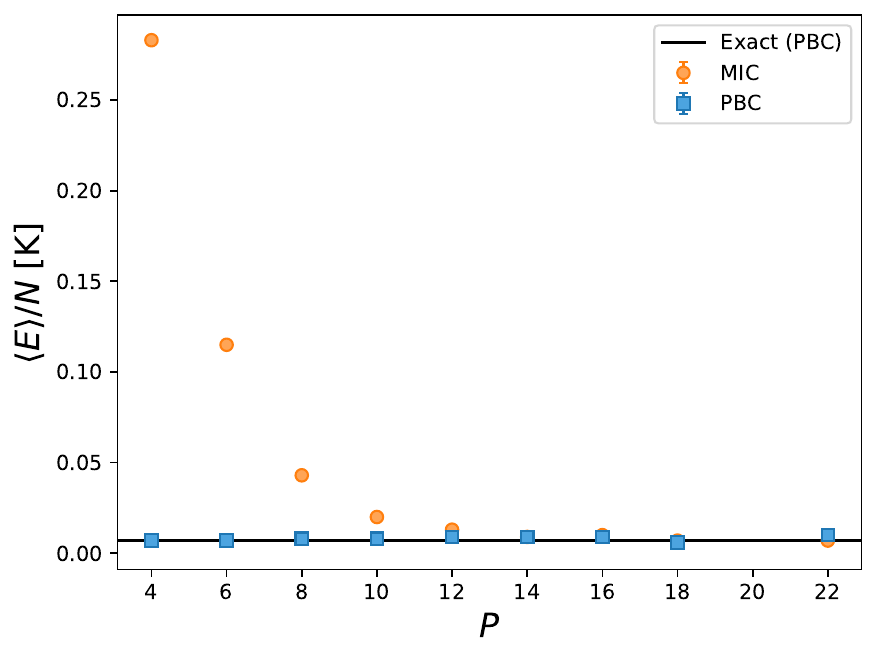}
    \caption{Energy per particle as a function of the number of beads $P$ for $N=64$ free bosons at $T=\temperature{0.5}$ and a density of $n=\density{0.035}$. The PBC algorithm converges at $P=4$ while the MIC approximation requires about $P=14$ beads.}
    \label{fig:winding-mic-convergence-free}
\end{figure}

\begin{figure}[ht]
    \centering
    \includegraphics[width=\linewidth]{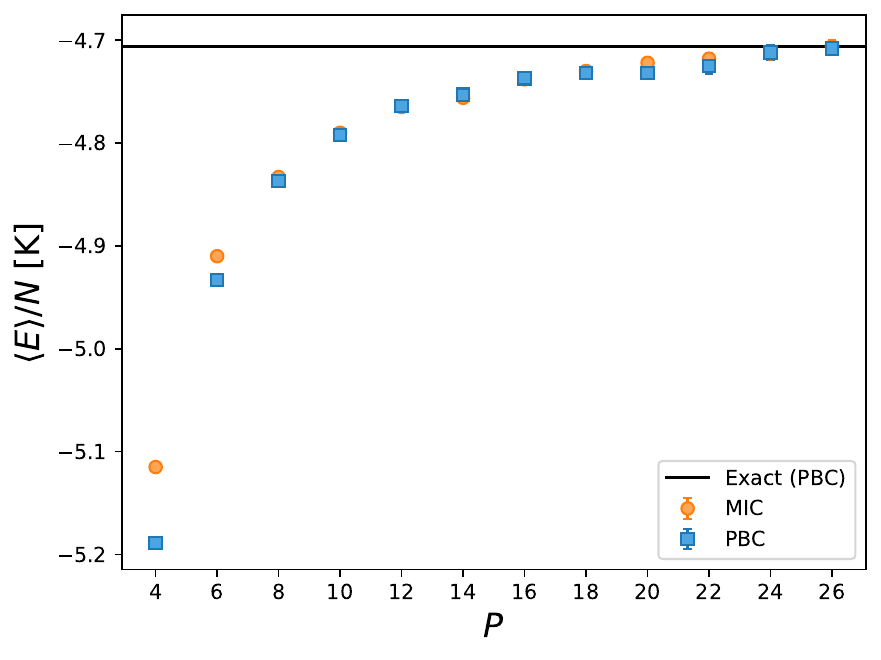}
    \caption{Energy per particle as a function of the number of beads $P$ for $N=32$ bosons in a sinusoidal trap at $T=\temperature{1.0}$ and density $n = \density{0.035}$. The difference between the PBC algorithm and MIC is negligible and both converge at the same $P$.
    }
    \label{fig:winding-mic-convergence-cosine-trap}
\end{figure}

To understand when MIC is a good approximation
, we consider for each spring the quantity
\begin{equation}
1 - \Prwindingmic{\ell}{j},
\label{eq:prob-discarded-mic}
\end{equation} 
where $\windingmic{\ell}{j}$ is the winding vector chosen by MIC accroding to~\Cref{eq:mic-def} for bead $j$ of particle $\ell$. \Cref{eq:prob-discarded-mic} is the total probability of all the winding vectors neglected by the MIC.
Notice that for $j = P$, $\Prwindingmic{\ell}{P}$ is a weighted sum over the possible connections of this bead. We now show that this quantity is indicative of whether MIC is a good approximation to the rigorous PBC treatment.

\Cref{fig:winding-mic-distribution-comparison} shows the discarded probability for bead $j=P$ of particle $\ell=1$ in a simulation of the free Bose gas at the same conditions as~\Cref{fig:winding-mic-convergence-free}. 
When $P=4$, the discarded probability is high during the simulation, i.e., there are multiple winding vectors with comparable, non-negligible probabilities, and the choice of a single winding vector by MIC neglects a significant share of important winding configurations. %
As $P$ increases, the discarded probability is lower throughout the simulation, and for $P=14$ it is nearly negligible throughout the entire simulation.
These results match \Cref{fig:winding-mic-convergence-free}, where the MIC converges at $P=14$, and differs significantly from the PBC results, which were already converged at $P=4$.

\Cref{fig:winding-mic-distribution-comparison-cosine} shows the discarded probability for the same bead but for the sinusoidal trap at the conditions of~\Cref{fig:winding-mic-convergence-cosine-trap}. The differences between MIC and PBC are relatively small, which is also evident from the discarded configurations. In $P=4$, the discarded probability is non-negligible but small, and indeed there are small differences for $P=4$ in~\Cref{fig:winding-mic-convergence-free}. When the discarded probability becomes negligible, as $P$ increases to $8$ and $14$, the small differences in energy disappear, in agreement with~\Cref{fig:winding-mic-convergence-free}.

The reason that the discarded probability decreases when $P$ increases can be understood from the probability distribution of winding numbers per~\Cref{eq:pr-winding,eq:windingpermute-joint-prob}. Larger $P$ makes the springs stiffer, leading to a higher energetic penalty for stretching them. This leads to a narrower winding distribution, and to MIC converging similarly to the rigorous treatment of PBC.
The same considerations apply when the box length increases, as a larger $L$ more strongly penalizes stretched springs. %

\begin{figure}[ht]
    \centering
    \includegraphics[width=\linewidth]{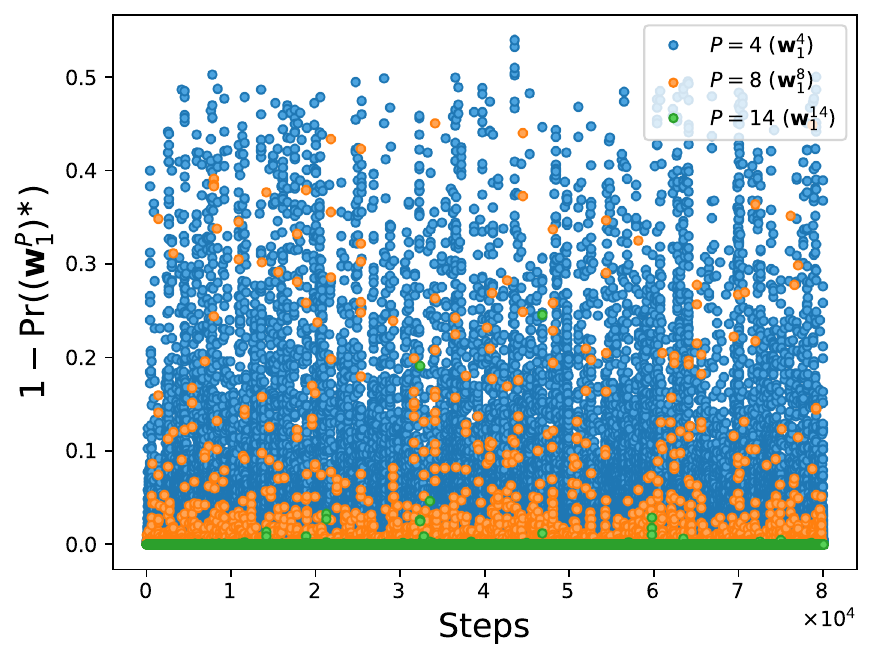}
    \caption{The probability of all the winding configurations for $\beadpos{1}{P}$ neglected by the MIC approximation during a simulation of $N=64$ free bosons at $T=\temperature{0.5}$, for different values of $P$. %
    As $P$ increases, the probability decays and the MIC approximation becomes better. At $P=14$ MIC and Winding yield practically identical trajectories, which is also when $\ev{E}/N$ in a MIC simulation converges to the correct result (see~\Cref{fig:winding-mic-convergence-free}).
    }
    \label{fig:winding-mic-distribution-comparison}
\end{figure}

\begin{figure}[ht]
    \centering
    \includegraphics[width=\linewidth]{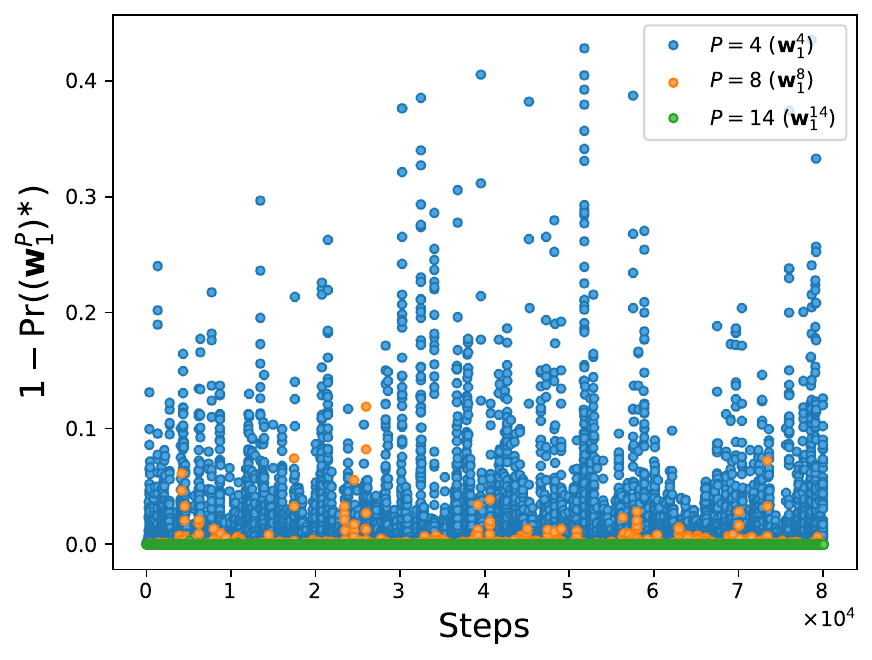}
    \caption{The probability of all the winding configurations neglected by the MIC approximation during a simulation of $N=32$ bosons in a sinusoidal trap at $T=\temperature{1.0}$, for different values of $P$. 
    As $P$ increases, the probability decays and the MIC approximation becomes better. At $P=14$ MIC and Winding yield practically identical trajectories, which is also when $\ev{E}/N$ in a MIC simulation converges to the correct result (see~\Cref{fig:winding-mic-convergence-cosine-trap}).}
    \label{fig:winding-mic-distribution-comparison-cosine}
\end{figure}

\section{Summary and conclusions}
In this paper, we developed an algorithm that rigorously accounts for PBC in bosonic PIMD simulations.
It required summing over the spring interaction between neighboring beads over all periodic images.
The key difficulty stemmed from the exponential number of periodic images (winding vectors) that had to be included, and that the winding vector of the last bead of every particle was coupled with the permutations, of which there is also an exponential number. We recently developed a quadratic scaling, recursive algorithm for handling the permutations but did not account for the sum over windings. Here, we showed that we can maintain the same recursive structure of the bosonic ring polymer potential and rigorously account for PBC by including the sum over windings in the cycle energies. 
We also evaluated the forces efficiently by deriving expressions for the connection and winding probabilities.
The resulting algorithm scales quadratically with the number of particles and linearly with the winding cutoff $\windingcap$.
We used the new algorithm to perform bosonic PIMD simulations with rigorous treatment of PBC for the first time, and benchmarked our approach on the free Bose gas and a model of cold atoms in a sinusoidal trap.
Finally, we carefully examined an approximate treatment of PBC using the MIC, derived a quantitative criterion to when is is a good approximation to the full periodic algorithm, and provided examples. We found that the MIC approximation becomes increasingly more suitable when the number of beads increases.
Our works enables simulations of bosonic systems with PIMD while accurately accounting for PBC.

\begin{acknowledgments}
B.H. acknowledges support by the USA-Israel Binational Science Foundation (grant No.\ 2020083) and the Israel Science Foundation (grants No.\ 1037/22 and 1312/22). Y.F. was supported by Schmidt Science Fellows, in partnership with the Rhodes Trust.
\end{acknowledgments}

\section*{Data Availability Statement}

The data that support the findings of this study, as well as details about the code, are openly available on GitHub, at \url{https://github.com/Hirshberg-Lab/bosonic-winding-pimd-data}.

\bibliography{refs}%

\begin{thebibliography}{28}%
\makeatletter
\providecommand \@ifxundefined [1]{%
 \@ifx{#1\undefined}
}%
\providecommand \@ifnum [1]{%
 \ifnum #1\expandafter \@firstoftwo
 \else \expandafter \@secondoftwo
 \fi
}%
\providecommand \@ifx [1]{%
 \ifx #1\expandafter \@firstoftwo
 \else \expandafter \@secondoftwo
 \fi
}%
\providecommand \natexlab [1]{#1}%
\providecommand \enquote  [1]{``#1''}%
\providecommand \bibnamefont  [1]{#1}%
\providecommand \bibfnamefont [1]{#1}%
\providecommand \citenamefont [1]{#1}%
\providecommand \href@noop [0]{\@secondoftwo}%
\providecommand \href [0]{\begingroup \@sanitize@url \@href}%
\providecommand \@href[1]{\@@startlink{#1}\@@href}%
\providecommand \@@href[1]{\endgroup#1\@@endlink}%
\providecommand \@sanitize@url [0]{\catcode `\\12\catcode `\$12\catcode
  `\&12\catcode `\#12\catcode `\^12\catcode `\_12\catcode `\%12\relax}%
\providecommand \@@startlink[1]{}%
\providecommand \@@endlink[0]{}%
\providecommand \url  [0]{\begingroup\@sanitize@url \@url }%
\providecommand \@url [1]{\endgroup\@href {#1}{\urlprefix }}%
\providecommand \urlprefix  [0]{URL }%
\providecommand \Eprint [0]{\href }%
\providecommand \doibase [0]{http://dx.doi.org/}%
\providecommand \selectlanguage [0]{\@gobble}%
\providecommand \bibinfo  [0]{\@secondoftwo}%
\providecommand \bibfield  [0]{\@secondoftwo}%
\providecommand \translation [1]{[#1]}%
\providecommand \BibitemOpen [0]{}%
\providecommand \bibitemStop [0]{}%
\providecommand \bibitemNoStop [0]{.\EOS\space}%
\providecommand \EOS [0]{\spacefactor3000\relax}%
\providecommand \BibitemShut  [1]{\csname bibitem#1\endcsname}%
\let\auto@bib@innerbib\@empty
\bibitem [{\citenamefont {Parrinello}\ and\ \citenamefont
  {Rahman}(1984)}]{doi:10.1063/1.446740}%
  \BibitemOpen
  \bibfield  {author} {\bibinfo {author} {\bibfnamefont {M.}~\bibnamefont
  {Parrinello}}\ and\ \bibinfo {author} {\bibfnamefont {A.}~\bibnamefont
  {Rahman}},\ }\href {\doibase 10.1063/1.446740} {\bibfield  {journal}
  {\bibinfo  {journal} {The Journal of Chemical Physics}\ }\textbf {\bibinfo
  {volume} {80}},\ \bibinfo {pages} {860} (\bibinfo {year} {1984})},\ \Eprint
  {http://arxiv.org/abs/https://doi.org/10.1063/1.446740}
  {https://doi.org/10.1063/1.446740} \BibitemShut {NoStop}%
\bibitem [{\citenamefont {Tuckerman}\ \emph {et~al.}(1993)\citenamefont
  {Tuckerman}, \citenamefont {Berne}, \citenamefont {Martyna},\ and\
  \citenamefont {Klein}}]{doi:10.1063/1.465188}%
  \BibitemOpen
  \bibfield  {author} {\bibinfo {author} {\bibfnamefont {M.~E.}\ \bibnamefont
  {Tuckerman}}, \bibinfo {author} {\bibfnamefont {B.~J.}\ \bibnamefont
  {Berne}}, \bibinfo {author} {\bibfnamefont {G.~J.}\ \bibnamefont {Martyna}},
  \ and\ \bibinfo {author} {\bibfnamefont {M.~L.}\ \bibnamefont {Klein}},\
  }\href {\doibase 10.1063/1.465188} {\bibfield  {journal} {\bibinfo  {journal}
  {The Journal of Chemical Physics}\ }\textbf {\bibinfo {volume} {99}},\
  \bibinfo {pages} {2796} (\bibinfo {year} {1993})},\ \Eprint
  {http://arxiv.org/abs/https://doi.org/10.1063/1.465188}
  {https://doi.org/10.1063/1.465188} \BibitemShut {NoStop}%
\bibitem [{\citenamefont {Markland}\ and\ \citenamefont
  {Ceriotti}(2018)}]{Markland2018}%
  \BibitemOpen
  \bibfield  {author} {\bibinfo {author} {\bibfnamefont {T.~E.}\ \bibnamefont
  {Markland}}\ and\ \bibinfo {author} {\bibfnamefont {M.}~\bibnamefont
  {Ceriotti}},\ }\href {\doibase 10.1038/s41570-017-0109} {\bibfield  {journal}
  {\bibinfo  {journal} {Nature Reviews Chemistry}\ }\textbf {\bibinfo {volume}
  {2}},\ \bibinfo {pages} {0109} (\bibinfo {year} {2018})}\BibitemShut
  {NoStop}%
\bibitem [{\citenamefont {Althorpe}(2021)}]{Althorpe2021}%
  \BibitemOpen
  \bibfield  {author} {\bibinfo {author} {\bibfnamefont {S.~C.}\ \bibnamefont
  {Althorpe}},\ }\href {\doibase 10.1140/epjb/s10051-021-00155-2} {\bibfield
  {journal} {\bibinfo  {journal} {The European Physical Journal B}\ }\textbf
  {\bibinfo {volume} {94}},\ \bibinfo {pages} {155} (\bibinfo {year}
  {2021})}\BibitemShut {NoStop}%
\bibitem [{\citenamefont {Hirshberg}\ \emph {et~al.}(2019)\citenamefont
  {Hirshberg}, \citenamefont {Rizzi},\ and\ \citenamefont
  {Parrinello}}]{hirshberg2019path}%
  \BibitemOpen
  \bibfield  {author} {\bibinfo {author} {\bibfnamefont {B.}~\bibnamefont
  {Hirshberg}}, \bibinfo {author} {\bibfnamefont {V.}~\bibnamefont {Rizzi}}, \
  and\ \bibinfo {author} {\bibfnamefont {M.}~\bibnamefont {Parrinello}},\
  }\href {\doibase 10.1073/pnas.1913365116} {\bibfield  {journal} {\bibinfo
  {journal} {Proceedings of the National Academy of Sciences}\ }\textbf
  {\bibinfo {volume} {116}},\ \bibinfo {pages} {21445} (\bibinfo {year}
  {2019})}\BibitemShut {NoStop}%
\bibitem [{\citenamefont {Feldman}\ and\ \citenamefont
  {Hirshberg}(2023)}]{quadratic-pidmb}%
  \BibitemOpen
  \bibfield  {author} {\bibinfo {author} {\bibfnamefont {Y.~M.~Y.}\
  \bibnamefont {Feldman}}\ and\ \bibinfo {author} {\bibfnamefont
  {B.}~\bibnamefont {Hirshberg}},\ }\href {\doibase 10.1063/5.0173749}
  {\bibfield  {journal} {\bibinfo  {journal} {The Journal of Chemical Physics}\
  }\textbf {\bibinfo {volume} {159}},\ \bibinfo {pages} {154107} (\bibinfo
  {year} {2023})}\BibitemShut {NoStop}%
\bibitem [{\citenamefont {Hirshberg}\ \emph {et~al.}(2020)\citenamefont
  {Hirshberg}, \citenamefont {Invernizzi},\ and\ \citenamefont
  {Parrinello}}]{doi:10.1063/5.0008720}%
  \BibitemOpen
  \bibfield  {author} {\bibinfo {author} {\bibfnamefont {B.}~\bibnamefont
  {Hirshberg}}, \bibinfo {author} {\bibfnamefont {M.}~\bibnamefont
  {Invernizzi}}, \ and\ \bibinfo {author} {\bibfnamefont {M.}~\bibnamefont
  {Parrinello}},\ }\href {\doibase 10.1063/5.0008720} {\bibfield  {journal}
  {\bibinfo  {journal} {The Journal of Chemical Physics}\ }\textbf {\bibinfo
  {volume} {152}},\ \bibinfo {pages} {171102} (\bibinfo {year} {2020})},\
  \Eprint {http://arxiv.org/abs/https://doi.org/10.1063/5.0008720}
  {https://doi.org/10.1063/5.0008720} \BibitemShut {NoStop}%
\bibitem [{\citenamefont {Dornheim}\ \emph {et~al.}(2020)\citenamefont
  {Dornheim}, \citenamefont {Invernizzi}, \citenamefont {Vorberger},\ and\
  \citenamefont {Hirshberg}}]{doi:10.1063/5.0030760}%
  \BibitemOpen
  \bibfield  {author} {\bibinfo {author} {\bibfnamefont {T.}~\bibnamefont
  {Dornheim}}, \bibinfo {author} {\bibfnamefont {M.}~\bibnamefont
  {Invernizzi}}, \bibinfo {author} {\bibfnamefont {J.}~\bibnamefont
  {Vorberger}}, \ and\ \bibinfo {author} {\bibfnamefont {B.}~\bibnamefont
  {Hirshberg}},\ }\href {\doibase 10.1063/5.0030760} {\bibfield  {journal}
  {\bibinfo  {journal} {The Journal of Chemical Physics}\ }\textbf {\bibinfo
  {volume} {153}},\ \bibinfo {pages} {234104} (\bibinfo {year} {2020})},\
  \Eprint {http://arxiv.org/abs/https://doi.org/10.1063/5.0030760}
  {https://doi.org/10.1063/5.0030760} \BibitemShut {NoStop}%
\bibitem [{\citenamefont {Xiong}\ and\ \citenamefont
  {Xiong}(2022{\natexlab{a}})}]{PhysRevE.106.025309}%
  \BibitemOpen
  \bibfield  {author} {\bibinfo {author} {\bibfnamefont {Y.}~\bibnamefont
  {Xiong}}\ and\ \bibinfo {author} {\bibfnamefont {H.}~\bibnamefont {Xiong}},\
  }\href {\doibase 10.1103/PhysRevE.106.025309} {\bibfield  {journal} {\bibinfo
   {journal} {Phys. Rev. E}\ }\textbf {\bibinfo {volume} {106}},\ \bibinfo
  {pages} {025309} (\bibinfo {year} {2022}{\natexlab{a}})}\BibitemShut
  {NoStop}%
\bibitem [{\citenamefont {Xiong}\ and\ \citenamefont
  {Xiong}(2022{\natexlab{b}})}]{doi:10.1063/5.0106067}%
  \BibitemOpen
  \bibfield  {author} {\bibinfo {author} {\bibfnamefont {Y.}~\bibnamefont
  {Xiong}}\ and\ \bibinfo {author} {\bibfnamefont {H.}~\bibnamefont {Xiong}},\
  }\href {\doibase 10.1063/5.0106067} {\bibfield  {journal} {\bibinfo
  {journal} {The Journal of Chemical Physics}\ }\textbf {\bibinfo {volume}
  {157}},\ \bibinfo {pages} {094112} (\bibinfo {year} {2022}{\natexlab{b}})},\
  \Eprint {http://arxiv.org/abs/https://doi.org/10.1063/5.0106067}
  {https://doi.org/10.1063/5.0106067} \BibitemShut {NoStop}%
\bibitem [{\citenamefont {Xiong}\ and\ \citenamefont
  {Xiong}(2022{\natexlab{c}})}]{doi:10.1063/5.0093472}%
  \BibitemOpen
  \bibfield  {author} {\bibinfo {author} {\bibfnamefont {Y.}~\bibnamefont
  {Xiong}}\ and\ \bibinfo {author} {\bibfnamefont {H.}~\bibnamefont {Xiong}},\
  }\href {\doibase 10.1063/5.0093472} {\bibfield  {journal} {\bibinfo
  {journal} {The Journal of Chemical Physics}\ }\textbf {\bibinfo {volume}
  {156}},\ \bibinfo {pages} {204117} (\bibinfo {year} {2022}{\natexlab{c}})},\
  \Eprint {http://arxiv.org/abs/https://doi.org/10.1063/5.0093472}
  {https://doi.org/10.1063/5.0093472} \BibitemShut {NoStop}%
\bibitem [{\citenamefont {Xiong}\ \emph {et~al.}(2024)\citenamefont {Xiong},
  \citenamefont {Liu},\ and\ \citenamefont {Xiong}}]{PhysRevE.110.065303}%
  \BibitemOpen
  \bibfield  {author} {\bibinfo {author} {\bibfnamefont {Y.}~\bibnamefont
  {Xiong}}, \bibinfo {author} {\bibfnamefont {S.}~\bibnamefont {Liu}}, \ and\
  \bibinfo {author} {\bibfnamefont {H.}~\bibnamefont {Xiong}},\ }\href
  {\doibase 10.1103/PhysRevE.110.065303} {\bibfield  {journal} {\bibinfo
  {journal} {Phys. Rev. E}\ }\textbf {\bibinfo {volume} {110}},\ \bibinfo
  {pages} {065303} (\bibinfo {year} {2024})}\BibitemShut {NoStop}%
\bibitem [{\citenamefont {Allen}\ and\ \citenamefont
  {Tildesley}(2017)}]{allen2017computer}%
  \BibitemOpen
  \bibfield  {author} {\bibinfo {author} {\bibfnamefont {M.~P.}\ \bibnamefont
  {Allen}}\ and\ \bibinfo {author} {\bibfnamefont {D.~J.}\ \bibnamefont
  {Tildesley}},\ }\href {https://books.google.co.il/books?id=nlExDwAAQBAJ}
  {\emph {\bibinfo {title} {Computer Simulation of Liquids}}},\ Oxford science
  publications\ (\bibinfo  {publisher} {Oxford University Press},\ \bibinfo
  {year} {2017})\BibitemShut {NoStop}%
\bibitem [{\citenamefont {Ceperley}(1995)}]{ceperley1995path}%
  \BibitemOpen
  \bibfield  {author} {\bibinfo {author} {\bibfnamefont {D.~M.}\ \bibnamefont
  {Ceperley}},\ }\href@noop {} {\bibfield  {journal} {\bibinfo  {journal}
  {Reviews of Modern Physics}\ }\textbf {\bibinfo {volume} {67}},\ \bibinfo
  {pages} {279} (\bibinfo {year} {1995})}\BibitemShut {NoStop}%
\bibitem [{\citenamefont {Marx}\ and\ \citenamefont
  {Müser}(1999)}]{Dominik_Marx_1999}%
  \BibitemOpen
  \bibfield  {author} {\bibinfo {author} {\bibfnamefont {D.}~\bibnamefont
  {Marx}}\ and\ \bibinfo {author} {\bibfnamefont {M.~H.}\ \bibnamefont
  {Müser}},\ }\href {\doibase 10.1088/0953-8984/11/11/003} {\bibfield
  {journal} {\bibinfo  {journal} {Journal of Physics: Condensed Matter}\
  }\textbf {\bibinfo {volume} {11}},\ \bibinfo {pages} {R117} (\bibinfo {year}
  {1999})}\BibitemShut {NoStop}%
\bibitem [{\citenamefont {Kleinert}(2009)}]{kleinert2009path}%
  \BibitemOpen
  \bibfield  {author} {\bibinfo {author} {\bibfnamefont {H.}~\bibnamefont
  {Kleinert}},\ }\href {https://books.google.co.il/books?id=VJ1qNz5xYzkC}
  {\emph {\bibinfo {title} {Path Integrals in Quantum Mechanics, Statistics,
  Polymer Physics, and Financial Markets}}},\ EBL-Schweitzer\ (\bibinfo
  {publisher} {World Scientific},\ \bibinfo {year} {2009})\BibitemShut
  {NoStop}%
\bibitem [{\citenamefont {Cao}(1994)}]{Cao1994}%
  \BibitemOpen
  \bibfield  {author} {\bibinfo {author} {\bibfnamefont {J.}~\bibnamefont
  {Cao}},\ }\href {\doibase 10.1103/physreve.49.882} {\bibfield  {journal}
  {\bibinfo  {journal} {Physical Review E}\ }\textbf {\bibinfo {volume} {49}},\
  \bibinfo {pages} {882–889} (\bibinfo {year} {1994})}\BibitemShut {NoStop}%
\bibitem [{\citenamefont {Marx}\ \emph {et~al.}(1993)\citenamefont {Marx},
  \citenamefont {Sengupta},\ and\ \citenamefont {Nielaba}}]{Marx1993}%
  \BibitemOpen
  \bibfield  {author} {\bibinfo {author} {\bibfnamefont {D.}~\bibnamefont
  {Marx}}, \bibinfo {author} {\bibfnamefont {S.}~\bibnamefont {Sengupta}}, \
  and\ \bibinfo {author} {\bibfnamefont {P.}~\bibnamefont {Nielaba}},\ }\href
  {\doibase 10.1063/1.466186} {\bibfield  {journal} {\bibinfo  {journal} {The
  Journal of Chemical Physics}\ }\textbf {\bibinfo {volume} {99}},\ \bibinfo
  {pages} {6031–6051} (\bibinfo {year} {1993})}\BibitemShut {NoStop}%
\bibitem [{\citenamefont {Miller}\ and\ \citenamefont
  {Clary}(2002)}]{Miller2002}%
  \BibitemOpen
  \bibfield  {author} {\bibinfo {author} {\bibfnamefont {T.~F.}\ \bibnamefont
  {Miller}}\ and\ \bibinfo {author} {\bibfnamefont {D.~C.}\ \bibnamefont
  {Clary}},\ }\href {\doibase 10.1063/1.1467342} {\bibfield  {journal}
  {\bibinfo  {journal} {The Journal of Chemical Physics}\ }\textbf {\bibinfo
  {volume} {116}},\ \bibinfo {pages} {8262–8269} (\bibinfo {year}
  {2002})}\BibitemShut {NoStop}%
\bibitem [{\citenamefont {Spada}\ \emph {et~al.}(2022)\citenamefont {Spada},
  \citenamefont {Giorgini},\ and\ \citenamefont {Pilati}}]{Spada2022}%
  \BibitemOpen
  \bibfield  {author} {\bibinfo {author} {\bibfnamefont {G.}~\bibnamefont
  {Spada}}, \bibinfo {author} {\bibfnamefont {S.}~\bibnamefont {Giorgini}}, \
  and\ \bibinfo {author} {\bibfnamefont {S.}~\bibnamefont {Pilati}},\ }\href
  {\doibase 10.3390/condmat7020030} {\bibfield  {journal} {\bibinfo  {journal}
  {Condensed Matter}\ }\textbf {\bibinfo {volume} {7}},\ \bibinfo {pages} {30}
  (\bibinfo {year} {2022})}\BibitemShut {NoStop}%
\bibitem [{\citenamefont {{Del Maestro}}\ \emph {et~al.}(2022)\citenamefont
  {{Del Maestro}}, \citenamefont {Nichols}, \citenamefont {Graves},\ and\
  \citenamefont {herdman}}]{adrian_del_maestro_2022_7271914}%
  \BibitemOpen
  \bibfield  {author} {\bibinfo {author} {\bibfnamefont {A.}~\bibnamefont {{Del
  Maestro}}}, \bibinfo {author} {\bibfnamefont {N.~S.}\ \bibnamefont
  {Nichols}}, \bibinfo {author} {\bibfnamefont {M.}~\bibnamefont {Graves}}, \
  and\ \bibinfo {author} {\bibnamefont {herdman}},\ }\href {\doibase
  10.5281/zenodo.7271914} {\enquote {\bibinfo {title} {Delmaestrogroup/pimc:
  Initial release},}\ } (\bibinfo {year} {2022})\BibitemShut {NoStop}%
\bibitem [{\citenamefont {Tuckerman}(2010)}]{tuckerman2010statistical}%
  \BibitemOpen
  \bibfield  {author} {\bibinfo {author} {\bibfnamefont {M.}~\bibnamefont
  {Tuckerman}},\ }\href@noop {} {\emph {\bibinfo {title} {Statistical
  mechanics: theory and molecular simulation}}}\ (\bibinfo  {publisher} {Oxford
  university press},\ \bibinfo {year} {2010})\BibitemShut {NoStop}%
\bibitem [{\citenamefont {Pollock}\ and\ \citenamefont
  {Ceperley}(1987)}]{PollockCeperley1987}%
  \BibitemOpen
  \bibfield  {author} {\bibinfo {author} {\bibfnamefont {E.~L.}\ \bibnamefont
  {Pollock}}\ and\ \bibinfo {author} {\bibfnamefont {D.~M.}\ \bibnamefont
  {Ceperley}},\ }\href {\doibase 10.1103/PhysRevB.36.8343} {\bibfield
  {journal} {\bibinfo  {journal} {Phys. Rev. B}\ }\textbf {\bibinfo {volume}
  {36}},\ \bibinfo {pages} {8343} (\bibinfo {year} {1987})}\BibitemShut
  {NoStop}%
\bibitem [{\citenamefont {Blanchard}\ \emph {et~al.}(2020)\citenamefont
  {Blanchard}, \citenamefont {Higham},\ and\ \citenamefont
  {Higham}}]{10.1093/imanum/draa038}%
  \BibitemOpen
  \bibfield  {author} {\bibinfo {author} {\bibfnamefont {P.}~\bibnamefont
  {Blanchard}}, \bibinfo {author} {\bibfnamefont {D.~J.}\ \bibnamefont
  {Higham}}, \ and\ \bibinfo {author} {\bibfnamefont {N.~J.}\ \bibnamefont
  {Higham}},\ }\href {\doibase 10.1093/imanum/draa038} {\bibfield  {journal}
  {\bibinfo  {journal} {IMA Journal of Numerical Analysis}\ }\textbf {\bibinfo
  {volume} {41}},\ \bibinfo {pages} {2311} (\bibinfo {year} {2020})},\ \Eprint
  {http://arxiv.org/abs/https://academic.oup.com/imajna/article-pdf/41/4/2311/40758053/draa038.pdf}
  {https://academic.oup.com/imajna/article-pdf/41/4/2311/40758053/draa038.pdf}
  \BibitemShut {NoStop}%
\bibitem [{\citenamefont {Kardar}(2007)}]{Kardar_2007}%
  \BibitemOpen
  \bibfield  {author} {\bibinfo {author} {\bibfnamefont {M.}~\bibnamefont
  {Kardar}},\ }\href@noop {} {\emph {\bibinfo {title} {Statistical Physics of
  Particles}}}\ (\bibinfo  {publisher} {Cambridge University Press},\ \bibinfo
  {year} {2007})\BibitemShut {NoStop}%
\bibitem [{\citenamefont {Borrmann}\ and\ \citenamefont
  {Franke}(1993)}]{Borrmann1993}%
  \BibitemOpen
  \bibfield  {author} {\bibinfo {author} {\bibfnamefont {P.}~\bibnamefont
  {Borrmann}}\ and\ \bibinfo {author} {\bibfnamefont {G.}~\bibnamefont
  {Franke}},\ }\href {\doibase 10.1063/1.464180} {\bibfield  {journal}
  {\bibinfo  {journal} {The Journal of Chemical Physics}\ }\textbf {\bibinfo
  {volume} {98}},\ \bibinfo {pages} {2484–2485} (\bibinfo {year}
  {1993})}\BibitemShut {NoStop}%
\bibitem [{\citenamefont {Krauth}(2006)}]{krauth2006statistical}%
  \BibitemOpen
  \bibfield  {author} {\bibinfo {author} {\bibfnamefont {W.}~\bibnamefont
  {Krauth}},\ }\href@noop {} {\emph {\bibinfo {title} {Statistical mechanics:
  algorithms and computations}}},\ Vol.~\bibinfo {volume} {13}\ (\bibinfo
  {publisher} {OUP Oxford},\ \bibinfo {year} {2006})\BibitemShut {NoStop}%
\bibitem [{\citenamefont {McLachlan}(1951)}]{McLachlan1951_Book}%
  \BibitemOpen
  \bibfield  {author} {\bibinfo {author} {\bibfnamefont {N.~W.}\ \bibnamefont
  {McLachlan}},\ }\href {https://books.google.co.il/books?id=xErX0AEACAAJ}
  {\emph {\bibinfo {title} {Theory and Application of Mathieu Functions}}}\
  (\bibinfo  {publisher} {Clarendon Press},\ \bibinfo {year}
  {1951})\BibitemShut {NoStop}%
\end{thebibliography}%

\iflong
\onecolumngrid
\clearpage
\appendix

\newcommand{\mainref}[2]{%
    \ifmaintextcompiled
        \Cref{#1}%
    \else
        {#2} in the main text%
    \fi
}

\newcommand{\mainrefplain}[2]{%
    \ifmaintextcompiled
        \Cref{#1}%
    \else
        {#2}%
    \fi
}

\newcommand{\maineqrefplain}[2]{\mainrefplain{#1}{Equation (#2)}}

\newcommand{\mainfigref}[2]{\mainref{#1}{Figure #2}}
\newcommand{\maineqref}[2]{\mainref{#1}{Equation (#2)}}
\newcommand{\mainsecref}[2]{\mainref{#1}{Section #2}}
\newcommand{\maineqrefs}[2]{\mainref{#1}{Equations #2}}

\newtheorem{theorem}{Theorem}
\newtheorem{lemma}{Lemma}
\newtheorem{definition}{Definition}
\newtheorem{corollary}{Corollary}[theorem]

\section{Additional details on the algorithm}
This section provides additional theoretical details on algorithm for bosonic PIMD with PBC and its quadratic scaling.

\subsection{Expression for the force on the first bead}
\label{sec:force-1-full}
In the main text we provide an expression for the force on the last bead of every particle
(\maineqrefplain{eq:force-on-P-sum-neighbors}{33}).
Here, we provide a similar expression for the force on the \emph{first} bead of every particle, as follows:
\begin{equation}
    \label{eq:force-on-1-sum-neighbors}
    \beadforce{\ell}{1}{\Vall} = 
    \sum_{\ellprime=1}^{N}{
        \sum_{\winding{\ellprime}{P}}{
        \Prrepandwinding{\ellprime}{\ell}{P} \cdot
        -\springconstant (\posbead{\ell}{1} - \posbead{\ellprime}{P} - \winding{\ellprime}{P}L)}
        }
    +
    \sum_{\winding{\ell}{1}}{
            \Prwinding{\ell}{1} \cdot
            -\springconstant (\posbead{\ell}{1} + \winding{\ell}{1} L - \posbead{\ell}{2} )}
    .
\end{equation}
In~\Cref{eq:force-on-1-sum-neighbors}, the first term corresponds to the contribution of the spring that connects the last bead of the previous particle to the first bead of particle $\ell$, which depends both on the permutation and the winding. The second term corresponds to the spring that connects the first bead of particle $\ell$ to the next bead in the same particle, and depends only on winding. Similarly to the force on the last bead, this expression allows to compute the force on the first bead of each particle in $\bigO(\windingcap N)$ time, and $\bigO(\windingcap N^2)$ overall for all such beads, based on the expressions for the winding and joint probabilities.

\subsection{Connection (marginal) probabilities}
\label{sec:connection-prob-full}
The probabilities $\Prrepnext{\ell}{\ellprime}$ are defined as the sum of the probabilities of all the configurations where the last bead of particle $\ell$ is connected to the first bead of particle $\ellprime$, with any choice of winding vectors. They can be computed according to the following expressions:
\begin{align}
\label{eq:connection-prob-start}
    \Prrepnext{\ell}{\ell+1} &= 
        1 - \frac{1}{\boltzmann{\Vall}} \boltzmann{\left(\Vto{\ell} + \Vfrom{\ell + 1}\right)}, &&
    \\
    \Prrepnext{\ell}{\ellprime} &= 
        \frac{1}{\ell} \frac{1}{\boltzmann{\Vall}} {\boltzmann{\left(\Vto{\ellprime-1} + \Efromto{\ellprime}{\ell} + \Vfrom{\ell+1}\right)}} && \mbox{for $\ellprime \leq \ell$},
    \\
    \Prrepnext{\ell}{\ellprime} &= 0 && \mbox{otherwise}.
\label{eq:connection-prob-end}
\end{align}
In these expressions, the potentials $\Vfrom{1},\ldots,\Vfrom{N}$ are defined by the recurrence relation
\begin{equation}
\label{eq:potentials-forward-si}
    \boltzmann{\Vfrom{u}} = \sum_{\ell=u}^{N}{\frac{1}{\ell} \boltzmann{\left(\Efromto{u}{\ell} + \Vfrom{\ell+1}\right)}}.
\end{equation}
Overall, the expressions for the connection probabilities are exactly as in the previous algorithm~\cite{quadratic-pidmb}, except that here the expression $\Efromto{\particledown}{\particleup}$ refers to the definition in the main text
(\maineqrefplain{eq:Ek}{25})
which sums over winding numbers.
The proof of these expressions appears in~\Cref{lem:direct-link-probability,lem:close-cycle-probability} in~\Cref{sec:proofs-connection-probabilities}.

\subsection{Sum over windings in more than one dimension}
\label{sec:multidimensional-winding-sum}
In this section, we explain how in our algorithm sums over winding vectors can be performed in time that is linear with the dimension.
In principle, for a maximal winding of $\windingcap$, the number of possible winding vectors $\winding{\ell}{j}$ is $(2\windingcap+1)^d$ --- a choice of winding number between $-\windingcap,\ldots,\windingcap$ for each of the $d$ coordinates of the bead. 
However, in our algorithm, all these sums can be computed in $\bigO(\windingcap)$ time instead, as follows.

\paragraph{Potential}
In the evaluation of the potential, we compute the total weight of all the windings of a given spring
(see~\maineqref{eq:edgewinding}{16}).
Explicitly writing the coordinates of the position and winding vectors, this can be simplified as follows:
\begin{align}
    \nonumber
    \edgewinding{u}{j}{v}{k} &= \sum_{\winding{u}{j}}{\boltzmann{\springenergyprefix \rdiffsquaredwinding{u}{j}{v}{j+1}}}
    \\
    \nonumber
    & = 
    \sum_{\winding{u}{j}}{
        \boltzmann{\sum_{i=1}^{d}{\springenergyprefix \left(\beadposdim{u}{j}{i} + \windingdim{u}{j}{i} - \beadposdim{v}{j+1}{i}\right)^2}}
    }
    \\
    \nonumber
    & = 
    \sum_{\winding{u}{j}}{
        \prod_{i=1}^{d}{\boltzmann{\springenergyprefix \left(\beadposdim{u}{j}{i} + \windingdim{u}{j}{i} - \beadposdim{v}{j+1}{i}\right)^2}}
    }
    \\
    & = 
    \prod_{i=1}^{d}{
        \sum_{\windingcomponent{u}{j}{i}}{\boltzmann{\springenergyprefix \left(\beadposdim{u}{j}{i} + \windingdim{u}{j}{i} - \beadposdim{v}{j+1}{i}\right)^2}}
    }.
\end{align}    
The sum per each coordinate is computed separately, in $\bigO(\windingcap)$ time, and then the product on the $d$ dimensions yields the desired value, in $\bigO(\windingcap)$ time ($d$ is considered a constant throughout).

\paragraph{Forces: Interior springs}
Consider, for example, the contribution of an interior spring between beads $j,j+1$ of particle $\ell$ to the %
force on bead $j$ of particle $\ell$. 
For the $i$th component, it takes the form
(see~\maineqref{eq:distinguiahble-bead-force-main-text}{21})

\begin{equation}
\label{eq:interior-spring-contribution-by-dim}
\sum_{\winding{\ell}{j}}{
        \Prwinding{\ell}{j} \cdot
        -\springconstant \left(\posbeaddim{\ell}{j}{i} + \windingdim{\ell}{j}{i}L - \posbeaddim{\ell}{j+1}{i}\right).
}
\end{equation}
We simplify the sum over winding vectors by noting that in an interior spring,
the probability of the total winding $\winding{\ell}{j}$, as shown in the main text
(\maineqrefplain{eq:pr-winding}{23}),
can be written as the product of the probabilities of its individual components: %
\begin{equation}
    \Prwinding{\ell}{j} = \Prwindingext{\ell}{j}{\ell}{j+1}
    = \prod_{i=1}^{d}{\left(
        \frac{\boltzmann{\springenergyprefix \rdiffsquaredwindingdim{\ell}{j}{\ell}{j+1}{i}}}
        {\left(\sum_{\windingdim{\ell}{j}{i}}{\boltzmann{\springenergyprefix \rdiffsquaredwindingdim{\ell}{j}{\ell}{j+1}{i}}}\right)}
    \right)}.
\end{equation}
Since the force expression depends only on the coordinate $\windingdim{\ell}{j}{i}$, it can be factored out, with the of the probabilities of the windings for the other coordinates summing to $1$. We thus obtain the following expression for~\Cref{eq:interior-spring-contribution-by-dim}:
\begin{equation}
    \sum_{\windingdim{\ell}{j}{i}}{
        \frac{\boltzmann{\springenergyprefix \rdiffsquaredwindingdim{\ell}{j}{\ell}{j+1}{i}}}
            {\left(\sum_{\windingdim{\ell}{j}{i}}{\boltzmann{\springenergyprefix \rdiffsquaredwindingdim{\ell}{j}{\ell}{j+1}{i}}}\right)}
        \cdot
        -\springconstant \left(\posbeaddim{\ell}{j}{i} + \windingdim{\ell}{j}{i} L - \posbeaddim{\ell}{j+1}{i}\right)
    }.
\end{equation}
In this expression, the sum is only over the windings of the specific coordinate and is performed in $\bigO(\windingcap)$ time.

\paragraph{Forces: Exterior springs}
The contribution of an exterior spring between bead $P$ of particle $\ell$ and bead $1$ of particle $\ellprime$ is, as shown in the main text
(see~\maineqrefplain{eq:force-on-P-sum-neighbors,eq:windingpermute-joint-prob}{33--34})
\begin{align}
    \nonumber
    &\qquad 
    \Prrepnext{\ell}{\ellprime} 
        \Prwindingext{\ell}{P}{\ellprime}{1}
        \cdot
        -\springconstant (\posbead{\ell}{P} + \winding{\ell}{P}L - \posbead{\ellprime}{1})
    .
\end{align}
The sums over windings in $\edgewinding{\ell}{P}{\ellprime}{1}$ can be evaluated in $\bigO(\windingcap)$ time exactly %
as in the case of interior springs.

\section{Thermodynamic kinetic energy estimator}
\label{sec:kinetic-estimator-full}
In this section, we derive the thermodynamic kinetic estimator for our bosonic PIMD algorithm with PBC.
As defined in the main text%
~(\maineqrefplain{eq:thermodynamics-estimator-main}{38}), the estimator is given by
\begin{equation}
\label{eq:thermodynamics-estimator-si}
    E = \frac{dPN}{2\beta} + \Vall +\beta\frac{\partial \Vall}{\partial\beta} +  \bar{U}.
\end{equation}

Since $\Vall$ is given by a recurrence relation, %
$\Vall = -\frac{1}{\beta} \ln{\frac{1}{N} \sum_{k=1}^{N}{\boltzmann{\left(\Vto{N-k} + \Enk{N}{k}\right)}}}$ (see~\maineqref{eq:our-forward-potential-recurrence}{24}),
the expression $\left(\Vall +\beta\frac{\partial \Vall}{\partial\beta}\right)$ is also given by a recurrence relation:
\begin{equation}
\label{eq:kinetic-est-recursion}
    \left(\Vto{\particleup} +\beta\frac{\partial \Vto{\particleup}}{\partial\beta}\right)
    =\frac{
        \sum\limits_{k=1}^{\particleup}
        \left[
        \left(\Vto{\particleup - k} + \beta\frac{\partial \Vto{\particleup - k}}{\partial\beta}\right) - \efromtodbeta{\particleup - k + 1}{\particleup} %
        \right]
        \boltzmann{\left(\Vto{\particleup - k}+\Efromto{\particleup-k+1}{\particleup}\right)}
    }
    {\sum\limits_{k=1}^{\particleup}{\boltzmann{\left(\Vto{\particleup - k} + \Efromto{\particleup - k + 1}{\particleup}\right)}}},
\end{equation}
with base case $\left(\Vto{0} +\beta\frac{\partial \Vto{0}}{\partial\beta}\right) = 0$. In~\Cref{eq:kinetic-est-recursion}, the term $\efromtodbeta{\particleup - k + 1}{\particleup}$ results from the derivative of 
a quantity
involving the cycle energy:
\begin{equation}
\label{eq:kinetic-e-deriv}
    \efromtodbeta{\particledown}{\particleup} = -\frac{\partial \left(\beta \Efromto{\particledown}{\particleup}\right)}{\partial \beta}.
\end{equation}

In the case of $\windingcap = 0$, this expression reduces to the thermodynamic estimator provided for the original algorithm~\cite{hirshberg2019path}.
 
To achieve an evaluation of the estimator with $\bigO(\windingcap (PN + N^2))$ scaling, we need to compute the terms $\efromtodbeta{\particledown}{\particleup}$ ($\particledown=1\ldots N, \particleup=\particledown\ldots N$) efficiently. Once these are known, evaluating the estimator through the recurrence is done in $\bigO(N^2)$ time (linear time for each $\left(\Vto{\particleup} +\beta\frac{\partial \Vto{\particleup}}{\partial\beta}\right)$, with $\particleup=1,\ldots,N$). %

To this end, we employ the following recurrence relation, %
\begin{equation}
\label{eq:edbeta-recurrence}
\begin{aligned}
    \efromtodbeta{\particledown}{\particleup}
    =
        \efromtodbeta{\particledown+1}{\particleup} 
        & -\sum_{\winding{\particleup}{P}}{
            \Prwindingext{\particleup}{P}{\particledown+1}{1}
            \cdot 
            \springenergyprefix \rdiffsquaredwinding{\particleup}{P}{\particledown+1}{1}
            }
            \\
        & + \sum_{\winding{\particledown}{P}}{
            \Prwindingext{\particledown}{P}{\particledown+1}{1}
            \cdot
            \springenergyprefix \rdiffsquaredwinding{\particledown}{P}{\particledown+1}{1}
        }
        \\
        & + \Einteriordbeta{\particledown}
        \\
        & + \sum_{\winding{\particleup}{P}}{
            \Prwindingext{\particleup}{P}{\particledown}{1}
            \cdot
            \springenergyprefix \rdiffsquaredwinding{\particleup}{P}{\particledown}{1}.
        }
\end{aligned}
\end{equation}
In~\Cref{eq:edbeta-recurrence}, $\Einteriordbeta{\particledown}$ is the contribution of the interior springs, %
\begin{equation}
\label{eq:edbeta-interior}
    \Einteriordbeta{\ell} = \sum_{j=1}^{P-1}{\sum_{\winding{\ell}{j}}{
        \Prwindingextshortdistinguishable{\ell}{j}
        \cdot \springenergyprefix \rdiffsquaredwinding{\ell}{j}{\ell}{j+1}
    }}.
\end{equation}
The base case of the recursion is the case $\efromtodbeta{\particleup}{\particleup}$, which is computed
\begin{equation}
\label{eq:edbeta-recurrence-base}
    \efromtodbeta{\particleup}{\particleup} = 
        \Einteriordbeta{\particledown}
        +
        \sum_{\winding{\particleup}{P}}
        {
            \Prwindingext{\particleup}{P}{\particleup}{1}
            \cdot\
            \springenergyprefix \rdiffsquaredwinding{\particleup}{P}{\particleup}{1}
        }.
\end{equation}

The recurrence for $\efromtodbeta{\particledown}{\particleup}$ in~\Cref{eq:edbeta-recurrence,eq:edbeta-recurrence-base} follows the same structure as the recurrence for $\Efromto{\particledown}{\particleup}$ (see~\mainsecref{sec:e-recursion-full}{III B 1}), extending cycles one particle at a time. Here, when we add or remove a spring, instead of adding or removing the sum of spring energies for all windings as is done for $\Efromto{\particledown}{\particleup}$, we add or remove the \emph{average} of the spring energy \emph{weighted} according to %
the winding distribution of that spring. %

With these expressions, the thermodynamic kinetic energy estimator is computed in $\bigO(\windingcap (N^2 + PN))$ time: the contribution of the interior springs according to~\Cref{eq:edbeta-interior} takes $\bigO(\windingcap PN)$ (for each of the $N$ particles, a sum over %
the windings for each of the $P-1$ beads); the base case according to~\Cref{eq:edbeta-recurrence-base} is $\bigO(\windingcap N)$; extending a cycle in each recurrence step in~\Cref{eq:edbeta-recurrence} is $\bigO(\windingcap)$, which is performed $\bigO(N^2)$ times. Once these steps are completed, the rest of the estimator calculation 
takes additional $\bigO(N^2)$, leading to $\bigO(N^2 + PN)$ scaling overall.

\section{Algorithm pseudocode}
\label{sec:alg-pseudocode}

The pseudocode of the bosonic PIMD algorithm with PBC is presented in~\Cref{alg:bosonic-pimdb-pbc}. %
The code differs from the main text in a number of ways:

\paragraph{Minor changes}
In~\cref{ln:force-exterior-beads} of~\Cref{alg:bosonic-pimdb-pbc}, the sum over neighbors of the exterior beads is restricted to certain neighbors, as the connection probability is zero in other cases (see~\Cref{sec:connection-prob-full}).

\paragraph{Separating interior and exterior beads}
{In~\Cref{alg:bosonic-pimdb-pbc}, the definition of the cycle energies does not include the contribution of the interior springs, which is different from the definition in \maineqref{eq:Ek}{25}. The reason is that the interior springs contribute equally to all permutations, resulting in a constant factor in the probabilities of~\maineqref{eq:representative-distribution}{32}. Our code takes advantage of this to reduce synchronization between processors.

With this definition of the cycle energies, $\Vall$ does not include the contribution of the interior springs, and therefore 
the expression for the thermodynamic energy estimator (\Cref{sec:kinetic-estimator-full}) must include the energy of the interior springs through an additional term, the second term in the expression,
\begin{equation}
\label{eq:kinetic-est-full-separate-interior}
    E
    = \frac{dNP}{2\beta}
    -\springenergyprefix \sum_{\ell=1}^{N}{\sum_{j=1}^{P-1}{\sum_{\winding{\ell}{j}}{
        \Prwinding{\ell}{j} \cdot
        \rdiffsquaredwinding{\ell}{j}{\ell}{j+1}
    }}}
    +
    \left(\Vall +\beta\frac{\partial \Vall}{\partial\beta}\right)
    + \bar{U}.
\end{equation}
In~\Cref{eq:kinetic-est-full-separate-interior}, $\Prwinding{\ell}{j}$ is the probability of the winding of an interior spring, defined in
\maineqref{eq:pr-winding}{23}, and the evaluation of $\left(\Vall +\beta\frac{\partial \Vall}{\partial\beta}\right)$ changes by modifying \Cref{eq:edbeta-recurrence} not to include the term $\Einteriordbeta{\particledown}$. %

{
\alghidebottomrule{}
\begin{algorithm}[H]
\DontPrintSemicolon
\caption{Quadratic scaling bosonic PIMD with PBC\label{alg:bosonic-pimdb-pbc}}
\tcp{Compute weights for interior beads}
\For{$\ell=1 \dots N$}
{
\For{$j=2 \dots P-1$}
{
\For{$\winding{\ell}{j}$}
{
$
\edgewindingspecificdistinguishable{\ell}{j} = \boltzmann{\springenergyprefix \rdiffsquaredwinding{\ell}{j}{\ell}{j+1}}
$
}
$\edgewinding{\ell}{j}{\ell}{j+1} = \sum\limits_{\winding{\ell}{j}}{\edgewindingspecific{\ell}{j}{\ell}{j+1}}$
}
}
\tcp{Compute force on interior beads}
\For{$\ell=1 \dots N$}
{
\For{$j=2 \dots P-1$}
{
$
\begin{aligned}
\beadforce{\ell}{j}{\Vall}
= 
        &\sum_{\winding{\ell}{j-1}}{
            \frac{\edgewindingspecific{\ell}{j-1}{\ell}{j}}{\edgewinding{\ell}{j-1}{\ell}{j}}
            \cdot
            -\springconstant (\posbead{\ell}{j} - \posbead{\ell}{j-1} - \winding{\ell}{j-1}L)}
        \\
        +
        &\sum_{\winding{\ell}{j}}{
            \Prwindingextshortdistinguishable{\ell}{j}
            \cdot
            -\springconstant (\posbead{\ell}{j} + \winding{\ell}{j}L - \posbead{\ell}{j+1})}
\end{aligned}
$
}
}
\rememberlines
\end{algorithm}
}

\newpage
\setcounter{algocf}{0}
{
\alghidetoprule{}
\setlength{\interspacetitleruled}{0pt}%
\setlength{\algotitleheightrule}{0pt}%
\begin{algorithm}[H]
\DontPrintSemicolon
\resumenumbering
\tcp{Compute weights for exterior beads}
\For{$\ell=1 \dots N$}
{
\For{$\ellprime=1 \dots \ell + 1$}
{
\For{$\winding{\ell}{P}$}
{
$
\edgewindingspecific{\ell}{P}{\ellprime}{1} = \boltzmann{\springenergyprefix \rdiffsquaredwinding{\ell}{P}{\ellprime}{1}}
$
}
$\edgewinding{\ell}{P}{\ellprime}{1} = \sum\limits_{\winding{\ell}{P}}{\edgewindingspecific{\ell}{P}{\ellprime}{1}}$
}
}
\tcp{Compute cycle energies $\Efromto{\particledown - \particleup + 1}{\particledown}$ (without interior springs) %
}
\For{$v=1\dots N$}
{
$
\boltzmann{\Efromto{\particleup}{\particleup}}=
    \edgewinding{\particleup}{P}{\particleup}{1}
$
\;
\For{$\particledown = \particleup - 1 \dots 1$}
{
$
\begin{aligned}
\boltzmann{\Efromto{\particledown}{\particleup}} =
\boltzmann{\Efromto{\particledown + 1}{\particleup}} & 
/
    \edgewinding{\particleup}{P}{\particledown + 1}{1}
\\
 & 
\cdot
    \edgewinding{\particledown}{P}{\particledown + 1}{1}
\\
 & 
\cdot
    \edgewinding{\particleup}{P}{\particledown}{1}
\end{aligned}
$
}
}
\tcp{Compute potentials $\Vto{\particleup}$ (without interior springs)}
$\Vto{0} = 0$\;
\For{$\particleup = 1 \dots N$}
{
$
\boltzmann{\Vto{\particleup}} = \frac{1}{\particleup} \sum_{k=1}^{\particleup} 
    \boltzmann{
        \left(
            \Vto{\particleup - k} + \Efromto{\particleup - k + 1}{\particleup}
        \right)
    }
$
}
\tcp{Compute potentials $\Vfrom{\particledown}$ (without interior springs)}
$\Vfrom{N + 1} = 0$\;
\For{$\particledown = N \dots 1$}
{
$
\boltzmann{\Vfrom{\particledown}} = 
\sum_{\ell = \particledown}^{N}
\frac{1}{\ell} \boltzmann{\left(
    \Efromto{\particledown}{\ell} + \Vfrom{\ell + 1}
\right)
}
$
}
\tcp{Compute connection probabilities}
\For{$\ell = 1 \dots N-1$}
{
$
\Prrepnext{\ell}{\ell+1} = 
    1 - \frac{1}{\boltzmann{\Vall}} \boltzmann{\left(\Vto{\ell} + \Vfrom{\ell + 1}\right)}
$
}
\For{$\ellprime = 1 \dots N$}
{
\For{$\ell = \ellprime \dots N$}
{
$
\Prrepnext{\ell}{\ellprime} = 
        \frac{1}{\ell} \frac{1}{\boltzmann{\Vall}} {\boltzmann{\left(\Vto{\ell'-1} + \Efromto{\ell'}{\ell} + \Vfrom{\ell+1}\right)}} 
$
}
}
\tcp{Compute force on exterior beads}
\For{$\ell = 1 \dots N$}
{
$\label{ln:force-exterior-beads}
\begin{aligned}
    \beadforce{\ell}{1}{\Vall} &=  
      \sum_{\winding{\ell}{1}}
        {
            \frac{\edgewindingspecific{\ell}{1}{\ell}{2}}{\edgewinding{\ell}{1}{\ell}{2}}
            \cdot
            -\springconstant (\posbead{\ell}{1} + \winding{\ell}{1}L - \posbead{\ell}{2})}
        \\
    +
    &
    \sum_{\ellprime=\ell - 1}^{N}{
        \Prrepnext{\ellprime}{\ell}
        \sum_{\winding{\ellprime}{P}}
        {
            \Prwindingext{\ellprime}{P}{\ell}{1}
            \cdot
            -\springconstant (\posbead{\ell}{1} - \posbead{\ellprime}{P} - \winding{\ellprime}{P}L)
        }
     }\\
    \beadforce{\ell}{P}{\Vall} &=  
      \sum_{\winding{\ell}{P-1}}
        {
            \frac{\edgewindingspecific{\ell}{P-1}{\ell}{P}}{\edgewinding{\ell}{P-1}{\ell}{P}}
            \cdot
            -\springconstant (\posbead{\ell}{P} - \posbead{\ell}{P-1} - \winding{\ell}{P-1}L)}
        \\
    +
    &
    \sum_{\ellprime=1}^{\ell + 1}{
        \Prrepnext{\ell}{\ellprime}
        \sum_{\winding{\ell}{P}}
        {
        \Prwindingext{\ell}{P}{\ellprime}{1}
        \cdot
        -\springconstant (\posbead{\ell}{P} + \winding{\ell}{P}L - \posbead{\ellprime}{1})}
     }
\end{aligned}
$
}
\end{algorithm}
}
\setcounter{algocf}{1}  
\section{Additional simulations details}
\label{sec:numerical-additional}

\subsection{General computational details}
All simulations are performed in a three-dimensional cubic geometry at a number density of $\density{0.035}$, with a mass of $\massunit{4.0}$. Except for the scalability analysis w.r.t. $N$, we simulate $N=64$ particles in the free Bose gas case (box side of $L = \length{12.23}$), while, for the sinusoidal trap, we use $N=32$ particles (box side of $L = \length{9.71}$). 
For the sinusoidal trap we use the potential $V\left(r_i\right) = V_0 \cos \left(\frac{2\pi}{L} r_i\right)$, where $r_i$ is the spatial coordinate, for all $i=x,y,z$, with $L$ equal to the box side length and the field amplitude $V_0 = \unit{0.3}{meV}$. Periodic boundary conditions are applied to all sides of the cubic box. Whenever we ran our PBC algorithm, we did so with a winding cutoff of $\windingcap = 1$ (except for the scalability analysis w.r.t. $\windingcap$).

Particles are initialized in a grid-like arrangement within the unit cell, %
with beads of the same particle initialized at the same position.
The initial velocities are sampled from a Maxwell--Boltzmann distribution, without zeroing the center of mass motion. To correctly sample the canonical ensemble, we employ a simple %
Langevin thermostat attached to the Cartesian coordinates of each bead, with time steps as outlined in the tables below, and a friction coefficient of $\left(100\Delta t\right)^{-1}$.
Observables such as energy are recorded every $100$ MD steps. When averaging over instantaneous values of the observables, we discard the first $20\%$ of the recorded values to reduce the effect of the equilibration steps on the final results.

\subsection{Benchmark: Energy as a function of temperature}
For~\mainfigref{fig:free-particles-energy-temperature}{4} (the free Bose gas), we ran simulations at temperatures $T=\temperature{0.5}, \temperature{1.0}, \temperature{1.5}, \temperature{2.0}$. For each temperature, we verified convergence w.r.t. $P$.

\begin{table}[H]
\centering
\begin{tabular}{ccccc}
\toprule 
Temperature $\left[\text{K}\right]$ & Time step $\left[\text{fs}\right]$ & Trajectories & Steps $\left[10^{7}\right]$ & Converged $P$\tabularnewline
\midrule
\midrule 
$0.5$ & $4.0$ & $200$ & $1.0$ & $10$\tabularnewline
\midrule 
$1.0$ & $1.0$ & $20$ & $1.0$ & $4$\tabularnewline
\midrule 
$1.5$ & $2.0$ & $20$ & $1.0$ & $4$\tabularnewline
\midrule 
$2.0$ & $2.0$ & $20$ & $1.0$ & $4$\tabularnewline
\bottomrule
\end{tabular}
\caption{Simulation parameters used for the free Bose gas (\mainfigref{fig:free-particles-energy-temperature}{4}).}
\end{table}

For~\mainfigref{fig:cosine-trap-energy-temperature}{5} (sinusoidal trap), we ran simulations at temperatures $T=\temperature{1.0}, \temperature{2.0}, \temperature{3.0}, \temperature{4.0}$. For each temperature, we verified convergence w.r.t. $P$.

\begin{table}[H]
\centering
\begin{tabular}{ccccc}
\toprule 
Temperature $\left[\text{K}\right]$ & Time step $\left[\text{fs}\right]$ & Trajectories & Steps $\left[10^{7}\right]$ & Converged $P$\tabularnewline
\midrule
\midrule 
$1.0$ & $2.0$ & $20$ & $1.0$ & $26$\tabularnewline
\midrule 
$2.0$ & $0.5$ & $20$ & $10.0$ & $14$\tabularnewline
\midrule 
$3.0$ & $0.3$ & $20$ & $1.0$ & $8$\tabularnewline
\midrule 
$4.0$ & $0.5$ & $20$ & $1.0$ & $6$\tabularnewline
\bottomrule
\end{tabular}
\caption{Simulation parameters used for the sinusoidal trap (\mainfigref{fig:cosine-trap-energy-temperature}{5}).}
\end{table}

\subsection{Scaling analysis}
All simulations pertaining to the scalability analysis were performed on a cluster of servers, each with two Intel Xeon Platinum 9242 CPU @ 2.30GHz, 386GB RAM, and a total of 96 cores. %
A single server was used for each measurement.
$P$ cores were used for each simulation (based on the replica parallelization mechanism, implemented with OpenMPI). %

For~\mainfigref{fig:scaling-size}{6} ($N$ scalability) we ran PIMD simulations of $N=64, 128, 256, 512, 1024$ free bosons with $P=32$, at temperature $T=\temperature{3.0}$ and an MD time step of $\Delta t = \dt{0.5}$. We ran a total of $10^3$ steps.

For~\mainfigref{fig:scaling-wind}{7} (winding cutoff scalability) we ran PIMD simulations of $N=2$ free bosons with $P=4$, at temperature $T=\temperature{3.0}$ and an MD time step of $\Delta t = \dt{0.5}$. We ran a total of $10^3$ steps, and we did so for $\windingcap = 64, 128, 256, 512, 1024$.
\subsection{Winding vs.\ PBC: Energy as a function of the number of beads}
For~\mainfigref{fig:winding-mic-convergence-free}{8} (the free Bose gas $\ev{E}/N$ convergence w.r.t. $P$), we ran simulations at temperature $T=\temperature{0.5}$, with a time step of $\Delta t = \dt{4.0}$ and $10^7$ steps overall. The result was averaged over $100$ independent trajectories.

For~\mainfigref{fig:winding-mic-convergence-cosine-trap}{9} (sinusoidal trap $\ev{E}/N$ convergence w.r.t. $P$), we ran simulations at temperature $T=\temperature{1.0}$, with a time step of $\Delta t = \dt{2.0}$ and $10^7$ steps overall. The result was averaged over $20$ independent trajectories.

\subsection{Winding vs.\ PBC: Discarded winding probability as a function of the simulation step}

For~\mainfigref{fig:winding-mic-distribution-comparison}{10} (discarded winding probability in the free Bose gas), we ran simulations of $N=64$ bosons at a temperature $T=\temperature{0.5}$, with the time step $\Delta t = \dt{4.0}$ and for $10^7$ steps. We did so for $P=4, 8, 14$ and extracted the discarded winding probability for the last bead of particle $1$.

Similarly, for~\mainfigref{fig:winding-mic-distribution-comparison-cosine}{11} (discarded winding probability in a sinusoidal trap), we ran simulations of $N=32$ bosons at a temperature $T=\temperature{1.0}$, with the time step $\Delta t = \dt{2.0}$. The rest of the parameters are exactly the same as in the free Bose gas case.

\section{Analytical results for the free Bose gas and the non-interacting sinusoidal trap}
In this section, we derive the analytical results for the energy of the free Bose gas and bosons in a sinusoidal trap without interaction to which we compare our numerical results%
~(\mainsecref{sec:empirical}{III C}).

\subsection{The free Bose gas}
\label{sec:analytical-free-particles}
The energy eigenvalues of a free particle in periodic boundary conditions in one dimension are~\cite{Kardar_2007}
\begin{equation*}
    E_n=\frac{2\pi^{2}\hbar^{2}}{mL^{2}}n^{2}, \quad n = 0, \pm 1, \pm 2, \dots \; ,
\end{equation*}
%
and the partition function is
\begin{equation}
\label{eq:free-z}
    Z_1(\beta) = \Tr(\boltzmann{\hat{H}})
        = \sum_{n} \boltzmann{E_n}
        = \sum_{n=-\infty}^{\infty} 
            \boltzmann{
                \frac{2\pi^{2}\hbar^{2}}{mL^{2}} n^{2}
            }
        = \vartheta_{3}
            \left(
                0, \boltzmann{
                        \frac{2\pi^{2}\hbar^{2}}{mL^{2}}
                    }
            \right)
    ,
\end{equation}
where $\vartheta_{3}$ is the elliptic theta function. In three dimensions, the partition function is modified from~\Cref{eq:free-z} to $\left(Z_1(\beta)\right)^3$.

Once the single-particle partition function is known, the $N$-particle bosonic partition function is determined~\cite{Borrmann1993} using the recurrence relation $Z_N=\sum_{k=1}^{N} z_k Z_{N-k}$, where $z_k=Z_1(k \beta)$ is the single-particle partition function at inverse temperature $k\beta$.  %
Based on this formula, the internal energy can also be evaluated recursively using\cite{krauth2006statistical}
\begin{equation}
\label{eq:noninteracting-energy-recurrence}
    \ev{E} = -\frac{1}{N Z_N} \sum_{k=1}^{N}
        \left(
            \pdv{z_k}{\beta} Z_{N-k} + 
            z_k \pdv{Z_{N-k}}{\beta}
        \right)
    .
\end{equation}
We evaluated the derivatives $\frac{\partial z_k}{\partial \beta}$ %
by finite differences.

\subsection{Sinusoidal trap}
\label{sec:analytical-cosine-trap}
The wave function of a particle in a box of side length $L$, with periodic boundary conditions, in a periodic external field $V(x)=V_0 \cos \left( k x \right)$ where $k=2\pi / L$, obeys the time-independent Schr\"{o}dinger equation
\begin{equation*}
    -\frac{\hbar^{2}}{2m}
    \dv[2]{\psi\left(x\right)}{x}
    + V_0 \cos\left(k x\right) \psi \left(x\right)
    = E \psi\left(x\right).
\end{equation*}
By making the substitution $y = kx / 2$ and rearranging terms we get \emph{Mathieu's differential equation}:
\begin{equation*}
    \dv[2]{\psi\left(x\right)}{y}
    +\left[a-2q\cos\left(2y\right)\right]\psi=0,
\end{equation*}
where
\begin{equation*}
    a=\frac{8mE}{\hbar^{2}k^{2}}, \quad q=\frac{4mV_0}{\hbar^{2}k^{2}}.
\end{equation*}
Since $x$ has periodic boundary conditions with period $2\pi/L$, the substituted $y$ has periodic boundary conditions of period $\pi$.
Such solutions exist only for specific values of $a$ and $q$. For a given $q \in \mathbb{R}$, there are infinitely many values of $a$ that yield periodic solutions, and hence an infinite number of discrete energy levels. The different values of $a$ are called \emph{characteristic numbers}, and they are typically listed as two sequences, $a_n(q)$ and $b_n(q)$.~\cite{McLachlan1951_Book}
For $q>0$ (that is, $V_0 > 0$), the values $a_n$ and $b_n$ are conventionally ordered %
\begin{equation*}
    a_0 < b_1 < a_1 < b_2 < a_2 < \dots
\end{equation*}
and the corresponding energy levels $E_0 < E_1 < \dots$ are 
\begin{equation*}
    \frac{\hbar^{2}k^{2}}{8m}a_0 < \frac{\hbar^{2}k^{2}}{8m}b_1 < \frac{\hbar^{2}k^{2}}{8m}a_1 < \frac{\hbar^{2}k^{2}}{8m}b_2 < \frac{\hbar^{2}k^{2}}{8m}a_2 < \ldots.
\end{equation*}
We use \texttt{scipy} to evaluate the characteristic numbers, and evaluate the partition function $Z_1(\beta)$ numerically by summing over the lowest $7$ energy levels.
The derivatives are evaluated by a similar sum over energy levels, $\frac{\partial z_k}{\partial \beta} = kd\sum_{n}{E_n \exp{-k\beta E_n}}$ where $d$ is the dimension.
Then, the energy is computed as above by the recurrence relation of~\Cref{eq:noninteracting-energy-recurrence}. 
\section{Proofs}
\label{sec:derivations}
In this section, we provide derivations for parts of the distinguishable and bosonic PIMD algorithms with PBC: 
\begin{enumerate}
    \item The expression for the probabilities of winding vectors that were used to evaluate the forces in distinguishable PIMD with PBC (\Cref{sec:derivations-distinguishable});
    \item The correctness of the bosonic spring potential for sampling the bosonic partition function with PBC (\Cref{sec:deriviations-bosonic-correctness});
    \item The expressions for the probabilities of winding vectors and of the connection probabilities that were used to evaluate the forces in distinguishable PIMD with PBC (\Cref{sec:proofs-connection-probabilities}).
\end{enumerate}

\subsection{Distinguishable PIMD with PBC: Derivations}
\label{sec:derivations-distinguishable}
We derive the expression for the probability of a winding vector, which was used for evaluating the forces in distinguishable PIMD with PBC.
\begin{theorem}
\label{thm:distinguishable-prob-si}
In PIMD for distinguishable particles with PBC, at given positions, the probability of a winding configuration $\windingsequence$ %
~(\maineqref{eq:distinguishable-representative-distribution}{20}) satisfies
\begin{equation}
    \Predistpermwinding{\text{id}}{\windingsequence} = \prod_{\ell=1}^{N}{\prod_{j=1}^{P}{\Prwinding{\ell}{j}}},
\end{equation}
where 
\begin{equation}
    \Prwinding{\ell}{j} = 
    \Prwindingextshortdistinguishable{\ell}{j}.
\end{equation}
\end{theorem}
\begin{proof}
    The proof is identical to~\Cref{thm:interior-winding-probability} below (winding probability of interior beads in bosonic PIMD with PBC), except that there is no %
    sum over permutations
    in the distinguishable particle case---the only permutation is the identity, where each particle is a separate ring.
\end{proof}

\subsection{Bosonic PIMD with PBC: Derivations}

\subsubsection{Correctness of bosonic PIMD with PBC}
\label{sec:deriviations-bosonic-correctness}
Here we prove the correctness of using the bosonic spring potential with PBC, which was defined by a recurrence relation, for sampling the bosonic partition function with PBC.

\begin{theorem}
\label{thm:bosonic-pimd-correctness}
    The bosonic spring potential with PBC $\Vall$, defined by the recurrence relation 
    in~\maineqref{eq:our-forward-potential-recurrence}{24}
    samples the bosonic partition function of~\maineqref{eq:z-bosonic-pbc-explicit}{6}.
\end{theorem}
\begin{proof}
The proof uses the equivalent expression for the potential derived in~\Cref{thm:bosonic-potential-sum-representatives} below.
As in the original algorithm~\cite{quadratic-pidmb}, the representative permutation $\rep{\sigma}$ (defined in Ref.~\citenum{quadratic-pidmb}) has the same cycle structure as the permutation $\sigma$. Thus, for every choice of winding vectors for all beads $\windingsequence$, the configurations $(\sigma, \windingsequence)$ and $(\rep{\sigma}, \windingsequence)$ contribute the same to the partition function, namely %
\begin{equation*}
    \int_{D(\mathcal{V})}{
        d\pos_1 \ldots d\pos_N \, 
        \boltzmann{\Eperm{\sigma, \windingsequence}}
    }
    =
    \int_{D(\mathcal{V})}{
        d\pos_1 \ldots d\pos_N \, 
        \boltzmann{\Eperm{\rep{\sigma}, \windingsequence}}
    }.
\end{equation*}
Hence, recalling also that the physical potential $\bar{U}$ is unaffected by exchange,
\begin{align}
    \mathcal{Z}^{\text B} _{\text{PBC}}
    &\propto 
    \int_{D(\mathcal{V})}{
        d\pos_1 \ldots d\pos_N \, 
        \frac{1}{\fact{N}}
        \sum_{\sigma}{
            \sum_{\windingsequence}{
            \boltzmann{\left(\Eperm{\sigma,\windingsequence} + \bar{U}\right)}
        }
        }
    }
    \\
    &=
    \int_{D(\mathcal{V})}{
        d\pos_1 \ldots d\pos_N \, 
        \frac{1}{\fact{N}}
        \sum_{\sigma}{
            \sum_{\windingsequence}{
            \boltzmann{\left(\Eperm{\rep{\sigma},\windingsequence} + \bar{U}\right)}
        }
        }
    }.
    \\
    \intertext{Applying~\Cref{thm:bosonic-potential-sum-representatives}, this is exactly} %
    &=
    \int_{D(\mathcal{V})}{
        d\pos_1 \ldots d\pos_N \, 
        \boltzmann{\left(\Vall + \bar{U}\right)}
    },
\end{align}
as desired.
\end{proof}

\begin{theorem}
\label{thm:bosonic-potential-sum-representatives}
    The bosonic spring potential with PBC $\Vall$ defined by the recurrence relation
    (\maineqref{eq:our-forward-potential-recurrence}{24})
    is equivalently defined by 
    \begin{equation}
    \label{eq:sum-over-representatives-si}
        \boltzmann{\Vall} = 
        \frac{1}{\fact{N}} \sum_{\sigma}{
            \sum_{\windingsequence}{
            \boltzmann{\Eperm{\rep{\sigma}, \windingsequence}}
        }
        },
    \end{equation}
    where $\repsym$ is defined in Ref.~\citenum{quadratic-pidmb}.
\end{theorem}
\begin{proof}
    The proof follows the same structure as the proof in the original algorithm (Theorem 1 in the SI of Ref.~\citenum{quadratic-pidmb}), showing that the r.h.s.\ satisfies the recurrence equation that in the main text was used to define the potential: 
    $\boltzmann{\Vall} = \frac{1}{N} \sum_{k=1}^{N}{\boltzmann{\left(\Vto{N-k} + \Enk{N}{k}\right)}}$ and $\Vto{0} = 0$.
    
    Fix a specific configuration $\windingsequence$. Repeating the proof of Theorem 1 in the SI of Ref.~\citenum{quadratic-pidmb}---the only difference being that the displacement from $\beadpos{\ell}{P}$ to the next bead is modified with $\winding{\ell}{P} L$ in all expressions for the spring energy---shows that 
    \begin{equation}
        \boltzmann{\Vtospecificwinding{N}}
        =
        \frac{1}{\fact{N}} \sum_{\sigma}{
            \boltzmann{\Eperm{\rep{\sigma}, \windingsequence}}},
    \end{equation}
    where $\boltzmann{\Vtospecificwinding{N}}$ is defined by the recurrence relation
    \begin{equation}
    \label{eq:our-forward-recurrence-specific-windingsequence}
        \boltzmann{\Vtospecificwinding{N}} = \frac{1}{N} \sum_{k=1}^{N}{\boltzmann{\left(\Vtospecificwinding{N-k} + \Enkspecificwinding{N}{k}\right)}},
    \end{equation}
    and $\Vtospecificwinding{0} = 0$.
    In~\Cref{eq:our-forward-recurrence-specific-windingsequence}, $\Efromtospecificwinding{\particledown}{\particleup}$ is the energy of the cycle connecting particles $\particledown,\ldots,\particleup$ with the winding vectors $\windingsequence$: %
    \begin{equation}
    \label{eq:Ekspecificwinding}
        \Efromtospecificwinding{\particledown}{\particleup} = \springenergyprefix 
        \sum\limits_{\ell=\particledown}^{\particleup}{\sum\limits_{j=1}^{P}{
                \rdiffsquaredwinding{\ell}{j}{\ell}{j+1}
            }
        }.
    \end{equation}
    Taking a sum over windings, we have
    \begin{equation}
    \label{eq:proof-of-recurrence-relation-with-winding}
        \sum_{\windingsequence}{
            \boltzmann{\Vtospecificwinding{N}}
        }
        =
        \sum_{\windingsequence}{\frac{1}{\fact{N}} \sum_{\sigma}{
            \boltzmann{\Eperm{\rep{\sigma}, \windingsequence}}}}.
    \end{equation}
    
    From the definition of cycle energies in~\maineqref{eq:Ek}{25}, the ordinary cycle energies $\boltzmann{\Efromto{\particledown}{\particleup}}$ are obtained by summing $\boltzmann{\Efromtospecificwinding{\particledown}{\particleup}}$ over the choices of winding vectors for the beads in the cycle:
    \begin{equation}
    \label{eq:cycle-energies-from-sum-cycle-energies-windingsequence}
        \boltzmann{\Enk{N}{k}} = \sum_{\windingsequencefromto{N-k+1}{N}}{\boltzmann{\Enkspecificwinding{N}{k}}}.
    \end{equation}
    Hence, a similar relationship holds for the potentials $\boltzmann{\Vto{\particleup}}$, which is obtained by summing $\boltzmann{\Vtospecificwinding{\particleup}}$ over the relevant winding vectors:
    \begin{equation}
    \label{eq:forward-potential-from-sum-forward-potential-windingsequence}
        \boltzmann{\Vto{N-k}} = \sum_{\windingsequencefromto{1}{N-k}}{\boltzmann{\Vtospecificwinding{N-k}}}.
    \end{equation}
    \Cref{eq:forward-potential-from-sum-forward-potential-windingsequence} holds because from the recurrence relation \Cref{eq:our-forward-recurrence-specific-windingsequence} and summing over winding vectors yields
    \begin{align*}
        \sum_{\windingsequencefromto{1}{\particleup}}{\boltzmann{\Vtospecificwinding{\particleup}}} 
        &= 
        \frac{1}{N} \sum_{k=1}^{\particleup}{
            \sum_{\windingsequencefromto{1}{\particleup}}\boltzmann{\Vtospecificwinding{\particleup-k}}\boltzmann{\Enkspecificwinding{\particleup}{k}}
        }
        \\
        &=
        \frac{1}{N} \sum_{k=1}^{\particleup}{
            \sum_{\windingsequencefromto{1}{\particleup-k}}\boltzmann{\Vtospecificwinding{\particleup-k}}
            \sum_{\windingsequencefromto{\particleup-k+1}{\particleup}}\boltzmann{\Enkspecificwinding{\particleup}{k}}
        }
        \\
        &\underset{\mbox{\scriptsize Eq.~(\ref{eq:cycle-energies-from-sum-cycle-energies-windingsequence})}}{=}
        \frac{1}{N} \sum_{k=1}^{\particleup}{
            \sum_{\windingsequencefromto{1}{\particleup-k}}\boltzmann{\Vtospecificwinding{\particleup-k}}
            \boltzmann{\Enk{\particleup}{k}}
        }
        \\
        &=
        \frac{1}{N} \sum_{k=1}^{\particleup}{
            \boltzmann{\Vto{\particleup-k}}
            \boltzmann{\Enk{\particleup}{k}}
        },
    \end{align*}
    where the last equality uses an induction on $\particleup$.
    
    Plugging~\Cref{eq:forward-potential-from-sum-forward-potential-windingsequence} into~\Cref{eq:proof-of-recurrence-relation-with-winding} provides the desired result.
\end{proof}

\subsubsection{Connection probabilities}
\label{sec:proofs-connection-probabilities}
In this section, we show that the connection probabilities, as well as the recurrence for the partial potentials---used for evaluating the forces in bosonic PIMD with PBC---retain the same form as in the previous algorithm\cite{quadratic-pidmb}, except for the change in the cycle energies to sum over winding vectors, as explained in the main text. %

\begin{theorem}
\label{lem:direct-link-probability}
For every $1 \leq \ell < N$, %
\begin{equation*}
	\Prrepnext{\ell}{\ell + 1} 
    = 
    1 - \frac{1}{\boltzmann{\Vall}} \boltzmann{\left(\Vto{\ell} + \Vfrom{\ell + 1}\right)}.
\end{equation*}
\end{theorem}
\begin{proof}
    We use the same approach used in the proof of~\Cref{thm:bosonic-potential-sum-representatives}.
    With a specific choice of winding vectors $\windingsequence$ for all the beads, the proof of Theorem 2 in the SI of Ref.~\citenum{quadratic-pidmb} shows that
    \begin{equation}
        \Prrepnextspecificwinding{\ell}{\ell + 1} = 
        1 - \frac{1}{\boltzmann{\Vtospecificwinding{N}}} \boltzmann{\left(\Vtospecificwinding{\ell} + \Vfromspecificwinding{\ell + 1}\right)},
    \end{equation}
    where the potentials $\Vtospecificwinding{1}, \ldots, \Vtospecificwinding{N}$ are defined by the recurrence relation
    \begin{equation}
    \label{eq:potentials-forward-si-specificwinding}
        \boltzmann{\Vfromspecificwinding{u}} = \sum_{\ell=u}^{N}{\frac{1}{\ell} \boltzmann{\left(\Efromtospecificwinding{u}{\ell} + \Vfromspecificwinding{\ell+1}\right)}},
    \end{equation}
    and $\Vfromspecificwinding{N+1} = 0$, with $\Efromtospecificwinding{u}{\ell}$ defined in the proof of~\Cref{thm:bosonic-potential-sum-representatives} above (\Cref{eq:Ekspecificwinding}).

    To obtain the connection probability regardless of winding, we rely on the law of total probability to write
    \begin{equation}
    \label{eq:prob-condition-on-all-windings}
        \Prrepnext{\ell}{\ell + 1}
        =
        \sum_{\windingsequence}{\Prrepnextspecificwinding{\ell}{\ell + 1} \cdot \Pr(\windingsequence)}.
    \end{equation}
    Now the probability for a winding configuration is
    \begin{equation}
        \Pr(\windingsequence) = \frac{1}{\boltzmann{\Vall}} \frac{1}{\fact{N}} \sum_{\sigma}{\boltzmann{\Eperm{\rep{\sigma}, \windingsequence}}}
        = \frac{1}{\boltzmann{\Vall}}{\boltzmann{\Vtospecificwinding{N}}},
    \end{equation}
    where the last equality was shown in the proof of~\Cref{thm:bosonic-potential-sum-representatives} (\Cref{eq:forward-potential-from-sum-forward-potential-windingsequence}).
    Plugging this expression into~\Cref{eq:prob-condition-on-all-windings} and performing the sum over winding vectors yields
    \begin{equation*}
	\Prrepnext{\ell}{\ell + 1} 
    = 
    1 - \frac{1}{\boltzmann{\Vall}} \boltzmann{\left(\Vto{\ell} + \Vfrom{\ell + 1}\right)},
\end{equation*}
    provided that $\boltzmann{\Vto{\particleup}} = \sum_{\windingsequencefromto{1}{\particleup}}{\boltzmann{\Vtospecificwinding{\particleup}}}$. 
    To see that this indeed holds, we take a sum over windings on the recursion of~\Cref{eq:potentials-forward-si-specificwinding}, yielding
    \begin{align*}
       \sum_{\windingsequencefromto{\particledown}{N}}{\boltzmann{\Vfromspecificwinding{\particledown}}} 
       = 
       \sum_{\ell=u}^{N}{\frac{1}{\ell} 
            \sum_{\windingsequencefromto{\particledown}{\ell}}{\boltzmann{\Efromtospecificwinding{\particledown}{\ell}}
            \sum_{\windingsequencefromto{\ell+1}{N}}{\Vfromspecificwinding{\ell+1}}}}
        \underset{\mbox{\scriptsize Eq.~(\ref{eq:cycle-energies-from-sum-cycle-energies-windingsequence})}}{=}
            \sum_{\ell=u}^{N}{\frac{1}{\ell} 
            {{\Efromto{\particledown}{\ell}}
            \sum_{\windingsequencefromto{\ell+1}{N}}{\Vfromspecificwinding{\ell+1}}}}.
    \end{align*}
    Thus $\sum_{\windingsequencefromto{\particledown}{N}}{\boltzmann{\Vfromspecificwinding{\particledown}}}$ satisfies the same recurrence relation that was used to define $\boltzmann{\Vto{\particleup}}$ in~\Cref{eq:potentials-forward-si}, showing that they coincide. The claim follows.
\end{proof}

\begin{theorem}
\label{lem:close-cycle-probability}
For every $1 \leq \particledown \leq \ell \leq N$,
\begin{equation*}
 \Prrepnext{\ell}{\ell'} = 
        \frac{1}{\ell} 
        \frac{\boltzmann{\left(\Vto{\ell'-1} + \Efromto{\ell'}{\ell} + \Vfrom{\ell+1}\right)}}{\boltzmann{\Vall}} ,
\end{equation*}
\end{theorem}
\begin{proof}
    Follows based on the proof of Theorem 3 in the SI of Ref.~\citenum{quadratic-pidmb} the same way that the proof of~\Cref{lem:direct-link-probability} follows from the proof of Theorem 2 in the SI of Ref.~\citenum{quadratic-pidmb}.
\end{proof}

\subsection{Winding probability}

In this section, we prove the expressions for the winding probabilities in the bosonic algorithm, both for interior and exterior beads, which were used for the force evaluation.

From the definition in~\maineqref{eq:representative-distribution}{32} follows the expression for the probability of a particular winding vector:
\begin{equation}
    \Pr(\winding{\ell}{j}) 
    = 
    \sum_{\sigma} 
    \sum_{
        \windingsequence \setminus \winding{\ell}{j}
    } 
    \frac{
        \boltzmann{
            E^{
                \rep{\sigma}, 
                \windingsequence
            }
        }
    }{
    N!\cdot \boltzmann{\Vall}
    },
\end{equation}
where $\sum_{\windingsequence \setminus \winding{\ell}{j}}$ represents summation over all possible values of all winding vectors $\winding{\ellprime}{\jprime}$ where $\ellprime \neq \ell$ or $\jprime \neq j$.

\begin{theorem}
\label{thm:interior-winding-probability}
    The probability of a winding vector $\winding{\ell}{j}$ for an interior bead ($j \neq P$) is
\begin{equation}
    \Pr(\winding{\ell}{j}) 
    = 
    \frac{
        \boltzmann{
            \springenergyprefix \rdiffsquaredwinding{\ell}{j}{\ell}{j+1}
        }
    }{
        \sum_{\winding{\ell}{j}}
        \boltzmann{
            \springenergyprefix \rdiffsquaredwinding{\ell}{j}{\ell}{j+1}
        }
    },
\end{equation}
\end{theorem}
\begin{proof}%
First, note that
\begin{equation*}
    \Delta E^{\rep{\sigma}, \windingsequence} = E^{\rep{\sigma}, \windingsequence} - \springenergyprefix \rdiffsquaredwinding{\ell}{j}{\ell}{j+1},
\end{equation*}
is independent of $\winding{\ell}{j}$, by the definition of $E^{\rep{\sigma}, \windingsequence}$. 
Now, observe that the spring energy term involving $\winding{\ell}{j}$ does not depend on the winding vectors of other beads \emph{nor} the permutation (these only affect $\Delta E^{\rep{\sigma}, \windingsequence}$). Of course, this relies on the assumption that $\beadpos{\ell}{j}$ is an interior bead. Therefore, we can write
\begin{align*}
    \Pr(\winding{\ell}{j}) & 
    = 
    \frac{
        \boltzmann{
            \springenergyprefix \rdiffsquaredwinding{\ell}{j}{\ell}{j+1}
        }
        \sum_{\sigma}
        \sum_{
            \windingsequence \setminus \winding{\ell}{j}
        }
        \boltzmann{
            \Delta E^{
                \rep{\sigma}, 
                \windingsequence
            }
        }
    }{
        \sum_{
            \winding{\ell}{j}
        }
        \boltzmann{
            \springenergyprefix \rdiffsquaredwinding{\ell}{j}{\ell}{j+1}
        }
        \sum_{\sigma}
        \sum_{
            \windingsequence \setminus \winding{\ell}{j}
        }
        \boltzmann{
            \Delta E^{
                \rep{\sigma}, 
                \windingsequence
            }
        }
    }. \\
    \intertext{
    The terms that do not depend on $\winding{\ell}{j}$ cancel out, resulting in
    }
    &= 
        \frac{
        \boltzmann{
            \springenergyprefix \rdiffsquaredwinding{\ell}{j}{\ell}{j+1}
        }
    }{
        \sum_{
            \winding{\ell}{j}
        }
        \boltzmann{
            \springenergyprefix \rdiffsquaredwinding{\ell}{j}{\ell}{j+1}
        }
    }.
\end{align*}
\end{proof}

\begin{theorem}
\label{lemma:winding-probability-given-connection-same-spring}
The conditional probability $\Pr\left(\winding{\ell}{P} \mid \nextof{\ell}{\ellprime}{\rep{\sigma}}\right)$ can be expressed as
\begin{equation}
    \Pr\left(\winding{\ell}{P} \mid \nextof{\ell}{\ellprime}{\rep{\sigma}}\right) = 
    \frac{
        \boltzmann{
            \springenergyprefix \rdiffsquaredwinding{\ell}{P}{\ellprime}{1}
        }
    }{
    \sum_{
        \winding{\ell}{P} 
    } \boltzmann{
            \springenergyprefix \rdiffsquaredwinding{\ell}{P}{\ellprime}{1}
        }
    }.
\end{equation}
\end{theorem}
\begin{proof}%
The probability of a winding vector $\winding{\ell}{P}$ corresponding to bead $P$ of particle $\ell$, given the connection $\nextof{\ell}{\ellprime}{\rep{\sigma}}$, is %
\begin{equation}
    \Pr\left(\winding{\ell}{P} \mid \nextof{\ell}{\ellprime}{\rep{\sigma}}\right) = 
    \sum_{
        \substack{
            \sigma \mbox{ s.t. } \\
            \nextof{\ell}{\ellprime}{\rep{\sigma}}
        }
    } 
    \sum_{
        \windingsequence \setminus \winding{\ell}{P}
    } 
    \frac{
        \boltzmann{
            E^{
                \rep{\sigma},
                \windingsequence
            }
        }
    }{
        \sum\limits_{
            \substack{
                \sigma \mbox{ s.t. } \\
                \nextof{\ell}{\ellprime}{\rep{\sigma}}
            }
        } 
        \sum\limits_{
            \windingsequence
        } 
        \boltzmann{
            E^{
                \rep{\sigma}, \windingsequence
            }
        }
    },
\end{equation}
where the sum over $\sigma$ is performed only over configurations such that they have the chosen connectivity of $\rep{\sigma}(\ell) = \ellprime$.

The quantity
\begin{equation*}
    \Delta E^{\rep{\sigma}, \windingsequence} = E^{\rep{\sigma}, \windingsequence} - \springenergyprefix \rdiffsquaredwinding{\ell}{P}{\ellprime}{1},
\end{equation*}
is independent of $\winding{\ell}{P}$, by the definition of $E^{\rep{\sigma}, \windingsequence}$. Conversely, the spring energy term involving $\winding{\ell}{P}$ does not depend on the winding vectors of other beads. %
Although the spring energy associated with $\winding{\ell}{P}$ generally depends on the permutation, in this case, it does not, as we consider only the subset of permutations where the connection that affects this specific spring is %
fixed. Therefore, we can write
\begin{align*}
    \Pr\left(\winding{\ell}{P} \mid \nextof{\ell}{\ellprime}{\rep{\sigma}}\right) & 
    = 
    \frac{
        \boltzmann{
            \springenergyprefix \rdiffsquaredwinding{\ell}{P}{\ellprime}{1}
        }
        \sum\limits_{
            \substack{
                \sigma \mbox{ s.t. } \\
                \nextof{\ell}{\ellprime}{\rep{\sigma}}
            }
        } 
        \sum_{
            \windingsequence \setminus \winding{\ell}{P}
        }
        \boltzmann{
            \Delta E^{\rep{\sigma}, \windingsequence}
        }
    }{
        \sum_{
            \winding{\ell}{P}
        }
        \boltzmann{
            \springenergyprefix \rdiffsquaredwinding{\ell}{P}{\ellprime}{1}
        }
        \sum\limits_{
            \substack{
                \sigma \mbox{ s.t. } \\
                \nextof{\ell}{\ellprime}{\rep{\sigma}}
            }
        } 
        \sum_{
            \windingsequence \setminus \winding{\ell}{P}
        }
        \boltzmann{
            \Delta E^{\rep{\sigma}, \windingsequence}
        }
    },
    \\
    \intertext{
    The terms that do not depend on $\winding{\ell}{P}$ cancel out, resulting in
    }
    &= 
        \frac{
        \boltzmann{
            \springenergyprefix \rdiffsquaredwinding{\ell}{P}{\ellprime}{1}
        }
    }{
        \sum_{
            \winding{\ell}{P}
        }
        \boltzmann{
            \springenergyprefix \rdiffsquaredwinding{\ell}{P}{\ellprime}{1}
        }
    }.
\end{align*}

\end{proof}  \fi

\end{document}